\documentclass[11pt]{article}
\usepackage{amssymb}
\usepackage{amsmath}
\usepackage{amsfonts}
\usepackage{bbm}
\usepackage{mathrsfs}
\usepackage{tikz}
\usepackage{hyperref}
\usepackage{setspace}

\usepackage[letterpaper,margin=2.45cm]{geometry}

\newcommand{\C}{\mathbb{C}}
\newcommand{\R}{\mathbb{R}}

\newcommand{\E}{\mathbb{E}}
\newcommand{\N}{\mathbb{N}}

\newcommand{\Z}{\mathbb{Z}}

\newcommand{\Hil}{\mathcal{H}}

\newcommand{\ket}[1]{|{#1}\rangle}

\newtheorem{theorem}{Theorem}[section]
\newtheorem{remark}[theorem]{Remark}
\newtheorem{proposition}[theorem]{Proposition}
\newtheorem{lemma}[theorem]{Lemma}
\newtheorem{definition}[theorem]{Definition}
\newtheorem{corollary}[theorem]{Corollary}

\newenvironment{proof}{{\bf Proof:}}{\hfill$\square$\vskip.5cm}
\newenvironment{proofof}{}{\hfill$\square$\vskip.5cm}

\newcommand{\B}{\mathbb{B}}

\newcommand{\y}{\boldsymbol{y}}

\renewcommand{\k}{\boldsymbol{k}}
\renewcommand{\B}{\mathbb{B}}

\renewcommand{\k}{\boldsymbol{k}}

\renewcommand{\r}{\boldsymbol{r}}
\renewcommand{\E}{\mathfrak{E}}

\renewcommand{\k}{\boldsymbol{\kappa}}

\newcommand{\bR}{\boldsymbol{R}}

\title{Asymptotic Ferromagnetic Ordering of Energy Levels 
for the Heisenberg Model
on Large Boxes
}
\author{Bruno Nachtergaele${}^1$, Wolfgang Spitzer${}^2$ and Shannon Starr${}^3$\\
\small
${}^1$ University of California at Davis, 
Davis, CA 95616, USA\\
\small
${}^2$ FernUniversit\"at in Hagen, Fakult\"at f\"ur Mathematik und Informatik,\\
\small
LG Angewandte Stochastik, 58084 Hagen, Germany\\
\small
${}^3$ University of Alabama at Birmingham,
 Applied Mathematics, Birmingham, AL 35294--1170}
\date{2 September 2015}

\begin{document}


\pagestyle{headings}

\maketitle

\begin{abstract}
We prove a result for the spin-$1/2$ quantum Heisenberg ferromagnet on $d$-dimensional boxes $\{1,\dots,L\}^d \subset \mathbb{Z}^d$. 
For any $n$, if $L$ is large enough, the Hamiltonian satisfies: 
among all vectors whose total spin is at most $(L^d/2)-n$, the minimum energy is attained by a vector whose total spin is exactly $(L^d/2)-n$.

\vspace{8pt}
\noindent
{\small \bf Keywords:} Heisenberg model, quantum spin systems, simple exclusion process, ordering of energy levels,
Aldous ordering, spectral gap.
\vskip .2 cm
\noindent
{\small \bf MCS numbers:} 82B10, 81R05, 81R50.
\end{abstract}

\section{Introduction}

In this article, we prove an asymptotic result for the ferromagnetic Heisenberg model on  boxes.
%
We prove that the {\em ferromagnetic ordering of energy levels (FOEL)} property holds up to some level, if the box is large enough.

The {\em FOEL condition at level $n$} means the following.
Consider a $d$-dimensional box 
$\B^d(L)$, defined to be $\{1,\dots,L\}^d \subset \Z^d$.
The ferromagnetic Heisenberg Hamiltonian on $\B^d(L)$
commutes with total spin, and the ground state subspace is the 
total spin subspace for the maximum possible spin $s=\frac{1}{2}L^d$.
For any $n$, we may restrict the Hamiltonian to the subspace spanned by total spin eigenvectors whose total spin satisfies $s\leq \frac{1}{2}L^d-n$.
Then the minimum energy eigenvalue among vectors in this subspace is 
attained by an eigenvector whose total spin is  $s=\frac{1}{2}L^d-n$.
Roughly stated: lower energies are attained at higher spin.
We prove that, for fixed $n$ and $d$ in $\{1,2,\dots\}$,  there exists  $L_0(n,d) \in \{1,2,\dots\}$ such that FOEL-$n$ holds on $\B^d(L)$
for each $L\geq L_0(n,d)$.
\subsection*{Discussion}
This property was proposed in \cite{NachtergaeleSpitzerStarr} as a ferromagnetic version of a famous theorem
by Lieb and Mattis for antiferromagnets and ferrimagnets \cite{LiebMattis}.
Moreover, the authors proved that FOEL holds at level $n$ for all $n \in \{0,1,\dots\}$,
for any $d=1$ dimensional box/chain of sidelength $L\geq 2n$. This is an optimal result because total spin
can never be less than $0$.
The proof from \cite{NachtergaeleSpitzerStarr} did not extend to periodic boundary conditions.
It relied on a special basis, called the ``Hulth\`en bracket basis,'' defined and studied
in the context of Bethe ansatz solvable models by Temperley and Lieb \cite{TemperleyLieb}.
Due to the basis considered by Temperley and Lieb,
the result generalized to the XXZ model on 1-dimensional boxes, which possesses 
symmetry of the quantum group $\mathcal{U}_q(\mathfrak{sl}_2)$.
It also could be proved for higher spin $\mathrm{SU}(2)$ models, as in \cite{NachtergaeleStarr},
as well as higher-spin $\mathcal{U}_q(\mathfrak{sl}_2)$-symmetric models, generalizing the XXZ model
in 1-dimension. The latter was carried out in \cite{NachtergaeleNgStarr}.
The FOEL property in $1d$ also has some applications, such as a detailed study of droplets in the thermodynamic limit for the XXZ chain
\cite{NachtergaeleSpitzerStarr2}.

In higher dimensions, or even for periodic boundary conditions, the argument based on the Hulth\`en bracket
basis does not work.
In fact, if $L$ is even, and $L>4$, and one considers the 1-dimensional torus of length $L$ (i.e., the ring),
then there is numerical evidence that FOEL is violated when one takes the level to be $n=(L/2)-1$. See, for example, \cite{SpitzerStarrTran}.

On the other hand, a deep and interesting result of Caputo, Liggett and Richthammer \cite{CaputoLiggettRichthammer}
showed that FOEL at level $n=1$ holds for all graphs. This had been conjectured, previously, by Aldous,
with some evidence by Diaconis. 
Before \cite{CaputoLiggettRichthammer}, Handjani and Jungreis wrote an important paper
including historical perspective \cite{HandjaniJungreis}.
Caputo, Liggett and Richthammer proved Aldous's conjecture.
It is interesting 
to ask how one may generalize their results for $n>1$, in a general way for all graphs?
We do not address that question in this article, but there is important work in this direction \cite{Cesi,AlonKozma,AlonKozma2,Dieker}.

Before their major advance there were two closely related  papers by Conomos and one of the authors \cite{ConomosStarr}
and Morris \cite{Morris} essentially proving Aldous's conjecture for  large boxes.
These papers built on an earlier inductive argument idea of Koma and Nachtergaele \cite{KomaNachtergaele},
who had calculated the exact spectral gap for the XXZ model on 1-dimensional boxes/chains,
and Handjani and Jungreis \cite{HandjaniJungreis}, who had considered examples of graphs for which
Aldous's conjecture was provable at that time.
The argument in \cite{ConomosStarr} and \cite{Morris} is easier than Caputo, Liggett and Richthammer's proof in \cite{CaputoLiggettRichthammer}.
In the present paper we generalize this argument from $n=1$ to $n>1$.
On large boxes with few spin deviations $n$ a kinetic energy term forces the low energy states to ``spread out.''
Then, for such states, the model is well approximated by an ideal Bose gas of ``magnons.''

Two years ago, a major advance was made by Correggi, Giuliani and Seiringer \cite{CorreggiGiulianiSeiringer,CorreggiGiulianiSeiringer2}.
They proved that the free energy at low temperature is asymptotically the same as for the non-interacting Bose gas
on the lattice.
One motivation for their work is the important open problem of proving $\mathrm{SU}(2)$ symmetry breaking
at low temperature for the Heisenberg ferromagnet.
See
$$ 
\text{\small \url{http://web.math.princeton.edu/~aizenman/OpenProblems.iamp/9901.HeisenbergFerr.html}.}
$$
There is a relation between our result and theirs.
A lemma in their paper establishes that the minimum energy among all eigenvectors with total spin $s=\frac{1}{2}L^d-n$
is at least $Cn$ for some $C>0$, uniformly in $L$ and $n$ (relative to the ground state energy $0$).
Their constant $C$ is not sharp for small $n$, but their bound is uniform in $n$.
We will obtain a sharp bound for small $n$, but not uniformly.

\section{Set-up and Statement of the main result}

Consider the Heisenberg model on a  finite graph
$\mathscr{G} = (\mathscr{V},\mathscr{E})$ with vertex set $\mathscr{V}$
and edge set $\mathscr{E}$.
We denote a reference Hilbert space $\Hil \cong \C^2$ for a single spin with orthonormal basis $\ket{\uparrow}$,
$\ket{\downarrow}$.
Relative to this basis, the Pauli spin matrices $\sigma^{(1)}$, $\sigma^{(2)}$, $\sigma^{(3)}$, and raising/lowering operators are
$$
\sigma^{(1)}\, =\, \begin{bmatrix} 0 & 1 \\ 1 & 0 \end{bmatrix}\, ,\quad
\sigma^{(2)}\, =\, i\begin{bmatrix} 0 & -1 \\ 1 & 0 \end{bmatrix}\, ,\quad
\sigma^{(3)}\, =\, \begin{bmatrix} 1 & 0 \\ 0 & -1 \end{bmatrix}\, ,\quad
\sigma^{+}\, =\, \begin{bmatrix} 0 & 1 \\ 0 & 0 \end{bmatrix}\, ,\quad
\sigma^{-}\, =\, \begin{bmatrix} 0 & 0 \\ 1 & 0 \end{bmatrix}\, .
$$
The spin-$1/2$ spin matrices are
$S^{(a)} = {\frac{1}{2}}\, \sigma^{(a)}$, for $a \in \{1,2,3\}$, and
$S^{\pm} = \sigma^{\pm}\, =\, S^{(1)} \pm i S^{(2)}$.
\begin{remark}
\label{rem:adj}
For each $a\in\{1,2,3\}$, we have $(\sigma^{(a)})^*=\sigma^{(a)}$. Also, $(\sigma^+)^* = \sigma^-$.
\end{remark}

For each $x \in \mathscr{V}$, there is an isomorphic copy  $\Hil_x \cong \Hil$.
The total Hilbert space for the Heisenberg model on $\mathscr{V}$  is
$\Hil_{\mathscr{V}} = \bigotimes_{x \in \mathscr{V}} \Hil_x$.
For each $x \in \mathscr{V}$, and  $a\in\{1,2,3\}$, we let $S_x^{(a)}$ denote the operator on $\Hil_{\mathscr{V}}$
such that,
for simple tensor product vectors, $\bigotimes_{y \in \mathscr{V}} \psi_y$ (where $\psi_y \in \Hil_y$ for each $y \in \mathscr{V}$),
$$
S^{(a)}_x \bigotimes_{y \in \mathscr{V}} \psi_y\,
=\, \bigotimes_{y \in \mathscr{V}} \psi'_y\quad \text{ for }\quad
\psi'_y\, =\,
 \begin{cases}
S^{(a)} \psi_x & \text{ for $y=x$,}\\
\psi_y & \text{ for $y\in \mathscr{V} \setminus\{x\}$.}
\end{cases}
$$
We define $S^{\pm}_x = S^{(1)}_x \pm i S^{(2)}_x$, consistently.
The Heisenberg Hamiltonian $H_{\mathscr{G}}$ is the operator on $\Hil_{\mathscr{V}}$:
\begin{equation}
\label{eq:HamDef}
H_{\mathscr{G}}\, =\, \sum_{\{x,y\} \in \mathscr{E}} h_{xy}\, ,\
\text{ where } \
h_{xy}\, =\, \frac{1}{4}\, \mathbbm{1} - \sum_{a=1}^{3} S_x^{(a)} S_y^{(a)}\, 
=\, \frac{1}{4}\, \mathbbm{1} - S_x^{(3)} S_y^{(3)} - \frac{1}{2}\, S_x^+ S_y^- - \frac{1}{2}\, S_x^- S_y^+\, ,
\end{equation}
where $\mathbbm{1}$ denotes the identity operator on $\Hil_{\mathscr{V}}$.
The model is $\mathrm{SU}(2)$ invariant. 
We may define 
\begin{equation}
\label{eq:SpinTot}
S_{\mathscr{V}}^{(a)}\, =\, \sum\nolimits_{x \in \mathscr{V}} S_x^{(a)}\, ,\, \text{ for $a=1,2,3$, and }\
S_{\mathscr{V}}^{\pm}\, =\, \sum\nolimits_{x \in \mathscr{V}} S_x^{\pm}\, =\, S_{\mathscr{V}}^{(1)} \pm i S_{\mathscr{V}}^{(2)}\, ,
\end{equation}
Then $[H_{\mathscr{G}},S^{(a)}_{\mathscr{V}}]=0$ for $a\in\{1,2,3\}$
which also implies $[H_{\mathscr{G}},S^{\pm}_{\mathscr{V}}]=0$.

We recall basic $\mathrm{SU}(2)$ theory to state the main result.
By Remark \ref{rem:adj} and the definitions
$(S_{\mathscr{V}}^{(a)})^*=S_{\mathscr{V}}^{(a)}$  for $a \in \{1,2,3\}$, 
and $(S_{\mathscr{V}}^+)^* = S_{\mathscr{V}}^-$.
Let $\Omega_{\mathscr{V}}(\emptyset) \in \Hil_{\mathscr{V}}$ be the simple tensor product vector, which is the product of $\ket{\uparrow} \in \Hil_x$, for
each $x \in \mathscr{V}$. 
Hence, $\Omega_{\mathscr{V}}(\emptyset)\in \ker(S_x^{(3)}-\frac{1}{2}\mathbbm{1})$ for each $x \in \mathscr{V}$.
For each $n \in \{1,2,\dots\}$ and each $(x_1,\dots,x_n) \in \mathscr{V}^n$, let
\begin{equation}
\label{eq:OmegaDEFn}
\Omega^{(n)}_{\mathscr{V}}(x_1,\dots,x_n)\, 
\stackrel{\mathrm{def}}{:=}\, S_{x_1}^{-} \cdots S_{x_n}^- \Omega_{\mathscr{V}}(\emptyset)\, .
\end{equation}
Then, for any $\pi \in S_n$, we have $\Omega^{(n)}_{\mathscr{V}}(x_{\pi_1},\dots,x_{\pi_n}) = \Omega^{(n)}_{\mathscr{V}}(x_1,\dots,x_n)$
since $S_{x_1}^-$,\dots,$S_{x_n}^-$ commute.
Given any subset $X \subseteq \mathscr{V}$, let us define $\Omega_{\mathscr{V}}(X)$ as follows:
if $|X|=n$ for some $n \in \{1,\dots,|\mathscr{V}|\}$ and $X = \{x_1,\dots,x_n\}$,
then we define $\Omega_{\mathscr{V}}(X)$ to be  $\Omega^{(n)}_{\mathscr{V}}(x_1,\dots,x_n)$.
Then we may easily deduce:
\begin{itemize}
\item
For $x,y \in \mathscr{V}$, $[S_x^{(3)},S_y^{\pm}]=\pm \delta_{x,y}S_y^{\pm}$ and $[S_x^+,S_y^-]=2\delta_{x,y}S_x^{(3)}$.
\item 
Summing, $[S_{\mathscr{V}}^{(3)},S_x^{\pm}] = \pm S_x^{\pm}$ for each $x \in \mathscr{V}$.
Also,
$[S_{\mathscr{V}}^{(3)},S_{\mathscr{V}}^{\pm}]=\pm S_{\mathscr{V}}^{\pm}$ and $[S_{\mathscr{V}}^+,S_{\mathscr{V}}^-] = 2 S_{\mathscr{V}}^{(3)}$.
\item
For $X \subseteq \mathscr{V}$ and $x \in\mathscr{V}$,  $S_x^+ \Omega_{\mathscr{V}}(X) = \mathbf{1}_{X}(x) \Omega_{\mathscr{V}}(X\setminus\{x\})$ and $S_x^- \Omega_{\mathscr{V}}(X) = \mathbf{1}_{\mathscr{V}\setminus X}(x) \Omega_{\mathscr{V}}(X\cup\{x\})$.
\end{itemize}
For each $m \in \R$, we define the $m$-magnon subspace,
\begin{equation}
\label{eq:MagSubspaceDef}
\Hil_{\mathscr{V}}^{\mathrm{mag}}(m)\, \stackrel{\mathrm{def}}{:=}\, \ker(S^{(3)}_{\mathscr{V}}-M\cdot\mathbbm{1})\, ,\quad
\text{ for $M= \frac{1}{2}|\mathscr{V}|-m$.}
\end{equation}
We denote the {\em Casimir} operator by the symbol $\mathcal{C}_{\mathscr{V}} = \sum_{a=1}^{3} (S_{\mathscr{V}}^{(a)})^2 = (S_{\mathscr{V}}^{(3)})^2 + \frac{1}{2} S_{\mathscr{V}}^+ S_{\mathscr{V}}^-
+\frac{1}{2} S_{\mathscr{V}}^- S_{\mathscr{V}}^+$.
This is also called the total-spin operator, $\boldsymbol{S}_{\mathscr{V}}^2$.
For $n \in \R$, we define the $n$-spin deviate subspace
\begin{equation}
\label{eq:SpinSubspaceDef}
\Hil^{\mathrm{spin}}_{\mathscr{V}}(n) = \ker(\mathcal{C}_{\mathscr{V}}-s(s+1)\mathbbm{1})\, ,\quad
\text{ for $s=\frac{1}{2}|\mathscr{V}|-n$.}
\end{equation}
A calculation shows $[\mathcal{C}_{\mathscr{V}},S_{\mathscr{V}}^{(a)}]=0$
for $a \in \{1,2,3\}$. 
So we can consider simultaneous eigenspaces:
\begin{gather*}
\Hil_{\mathscr{V}}\, =\, \bigoplus_{(m,n) \in \mathcal{A}_{\mathscr{V}}} \Hil_{\mathscr{V}}(m,n)\, ,\ \text{where}\\
\forall (m,n) \in \R^2\, ,\
\Hil_{\mathscr{V}}(m,n)\, \stackrel{\mathrm{def}}{:=}\, \Hil^{\mathrm{mag}}_{\mathscr{V}}(m) \cap \Hil_{\mathscr{V}}^{\mathrm{spin}}(n)\, ,\
\text{and}\\
\mathcal{A}_{\mathscr{V}}\, \stackrel{\mathrm{def}}{:=}\, \{(m,n)\in \R^2\, :\, \Hil_{\mathscr{V}}(m,n)\neq \{0\}\}\, .
\end{gather*}
Since $\mathcal{C}_{\mathscr{V}}$ is in the subalgebra generated by $S_{\mathscr{V}}^{(a)}$, for $a \in \{1,2,3\}$,
we see that $H_{\mathscr{G}}$ also commutes with $\mathcal{C}_{\mathscr{V}}$.
Hence, $\Hil_{\mathscr{V}}(m,n)$ is an invariant subspace for $H_{\mathscr{G}}$ for each $(m,n) \in \mathcal{A}_{\mathscr{V}}$.
Also note that since $[\mathscr{C}_{\mathscr{V}},S_{\mathscr{V}}^{(a)}]=0$ for $a \in \{1,2,3\}$, then
$[\mathcal{C}_{\mathscr{V}},S_{\mathscr{V}}^{\pm}]=0$, by \ref{eq:SpinTot}).
The following lemma gives the structure of the $\mathrm{SU}(2)$ representation on $\Hil_{\mathscr{V}}$.
We will not prove this lemma, which is standard.
(See, e.g.,
\cite{Edmonds}.)

Let $\mathcal{P}(\mathscr{V})$ denote the power set of $\mathscr{V}$ and
let 
\begin{equation}
\label{eq:PowerCardm}
\forall m \in \{0,\dots,|\mathscr{V}|\}\, ,\qquad
\mathcal{P}_m(\mathscr{V})\, 
\stackrel{\mathrm{def}}{:=}\, \{X \in \mathcal{P}(\mathscr{V})\, :\, |X|=m\}\, .
\end{equation}
\begin{lemma} 
 \label{lem:preliminary}
(1) For $m \in \{0,\dots,|\mathscr{V}|\}$, we have
$\Hil_{\mathscr{V}}^{\mathrm{mag}}(m) = \operatorname{span}(\{\Omega_{\mathscr{V}}(X)\, :\, X\in\mathcal{P}_m(\mathscr{V})\})$,
and $\langle \Omega_{\mathscr{V}}(X),\Omega_{\mathscr{V}}(Y) \rangle = \delta_{X,Y}$
for $X,Y\in\mathcal{P}_m(\mathscr{V})\})$.
In particular, $\dim(\Hil_{\mathscr{V}}^{\mathrm{mag}}(m)) = \binom{|\mathscr{V}|}{m}$.\\[2pt]
(2) For $m \in \{0,\dots,|\mathscr{V}|-1\}$, we have $S^-_{\mathscr{V}} \Hil^{\mathrm{mag}}_{\mathscr{V}}(m) \subseteq \Hil^{\mathrm{mag}}_{\mathscr{V}}(m+1)$,
and $S^-_{\mathscr{V}} \Hil^{\mathrm{mag}}_{\mathscr{V}}(|\mathscr{V}|)=0$.\\[2pt]
(3)
Define $\mathcal{A}'(\mathscr{V})$ to be  $\{(m,n) \in \Z^2\, :\,
m\geq0\, ,\ n\geq 0\, ,\ n\leq m\, ,\ m+n\leq |\mathscr{V}|\}$. 
For each $(m,n) \in \mathcal{A}'(\mathscr{V})$,
define $\Hil'_{\mathscr{V}}(m,n) \subset \Hil_{\mathscr{V}}$
as
$\Hil'_{\mathscr{V}}(m,n)
 = (S_{\mathscr{V}}^-)^{m-n} 
\left(\Hil^{\mathrm{mag}}_{\mathscr{V}}(n) \cap \ker(S^+_{\mathscr{V}})\right)$,
where $(S_{\mathscr{V}}^-)^0=\mathbbm{1}$.
Then
$\Hil'_{\mathscr{V}}(m,n) \subseteq \Hil_{\mathscr{V}}(m,n)$.\\[2pt]
(4) 
For each $(m,n) \in \mathcal{A}'(\mathscr{V})$,
we have
$\Hil'_{\mathscr{V}}(m,n) \subseteq \ker(S_{\mathscr{V}}^+ S_{\mathscr{V}}^- -(s+M)(s-M+1)\mathbbm{1})$ where $s=\frac{1}{2}|\mathscr{V}|-n$, $M=\frac{1}{2}|\mathscr{V}|-m$.\\[2pt]
(5) For each $(m,n) \in \mathcal{A}'(\mathscr{V})$,
we have $\dim(\Hil'_{\mathscr{V}}(m,n)) = \binom{|\mathscr{V}|}{n}  - \binom{|\mathscr{V}|}{n-1}$, 
where $\binom{|\mathscr{V}|}{-1}$ is set to $0$.\\[2pt]
(6)
The combinatorial identity
$
\sum_{(m,n) \in \mathcal{A}'(\mathscr{V})} \dim(\Hil'_{\mathscr{V}}(m,n)) = \dim(\Hil_{\mathscr{V}})$
is true.
So $\mathcal{A}(\mathscr{V})$ is $\mathcal{A}'(\mathscr{V})$.
And, for each $(m,n) \in \mathcal{A}(\mathscr{V})$, we have that
$\Hil_{\mathscr{V}}(m,n)$ is $\Hil'_{\mathscr{V}}(m,n)$.
\end{lemma}
%
The following is key for us.
\begin{proposition}
\label{prop:LiebMattis}
$\min\operatorname{spec}(H_{\mathscr{G}})=0$.
Moreover, if $\mathscr{G}$
is connected, then 
$\operatorname{ker}(H_{\mathscr{G}})\, =\, 
\Hil_{\mathscr{V}}^{\mathrm{spin}}(0)$.
\end{proposition}
Proposition \ref{prop:LiebMattis} is one small part of a famous result of Lieb and Mattis on ``ordering of energy levels,''
 \cite{LiebMattis}.
We will not reprove Proposition \ref{prop:LiebMattis}.
But the reader is strongly recommended to consult \cite{LiebMattis} if not already acquainted with that article.
It provides the main motivation for the present article.

\begin{corollary}
\label{eq:psd}
For any finite graph $\mathscr{G} = (\mathscr{V},\mathscr{E})$, we have that $H_{\mathscr{G}}$ is positive semi-definite.
\end{corollary}
\begin{proof}
From Proposition \ref{prop:LiebMattis} all eigenvalues are nonnegative.
\end{proof}
For operators, $A\geq B$ means $A-B$ is positive semi-definite.
Hence, for later reference, we record:
\begin{equation}
\label{eq:psdHam}
H_{\mathscr{G}}\, \geq\, 0\, .
\end{equation}
We may denote by $\ell^2(\mathscr{V})$ the finite dimensional Hilbert space of all functions $F : \mathscr{V} \to \C$
with $\|F\|^2 = \sum_{x \in \mathscr{V}} |F(x)|^2$.
Then the graph Laplacian is an operator $-\Delta_{\mathscr{G}}$ on $\ell^2(\mathscr{V})$, given by the formula 
\begin{equation}
\label{eq:GraphLap}
-\Delta_{\mathscr{G}} F(x)\, =\, \frac{1}{2}\, \sum_{y \in \mathscr{N}(\mathscr{E},x)} (F(x)-F(y))\, ,
\end{equation}
for each $x \in \mathscr{V}$, where we define 
the neighborhood of $x$ as
\begin{equation}
\label{eq:ExNeighborhoos}
\mathscr{N}(\mathscr{E},x)\, \stackrel{\mathrm{def}}{:=}\, \{y \in \mathscr{V}\, :\, (x,y) \in \mathscr{E}\}\, .
\end{equation} 
Then, for any $F \in \ell^2(\mathscr{V})$, 
\begin{equation}
\label{eq:GraphLaplacianDirichletForm}
\langle F,-\Delta_{\mathscr{G}} F\rangle\, =\, 
\frac{1}{2}\, \sum_{\{x,y\} \in \mathscr{E}}  |F(x)-F(y)|^2\, .
\end{equation}
This is related to  (\ref{eq:psdHam})
as it seems was first observed by Toth \cite{Toth},
because
for each $n \in \{0,\dots,|\mathscr{V}|\}$ there is a graph derived from $\mathscr{G}$ 
such that $H_{\mathscr{G}} \restriction \Hil^{\mathrm{mag}}(n)$ is unitarily equivalent to the 
Laplacian for that graph, if one restricts to symmetric tensors, as one does for a Bose gas.

Let us now explain the ferromagnetic ordering of energy levels property, which we call FOEL, for short.
For each $n\in\{0,\dots,\lfloor \frac{1}{2}|\mathscr{V}| \rfloor\}$, we define $\E_n(\mathscr{G})$ as the minimum energy among $n$-spin deviates:
\begin{equation}
\label{eq:EDef}
\E_n(\mathscr{G})\, =\, \min\operatorname{spec}\big(H_{\mathscr{G}}
\restriction \Hil^{\mathrm{spin}}_{\mathscr{V}}(n)\big)\, =\,
\min\left\{\lambda \in \R\, :\,  
\Hil^{\mathrm{spin}}_{\mathscr{V}}(n) \cap \ker(H_{\mathscr{G}}-\lambda \mathbbm{1})
\neq \{0\}\right\}\, .
\end{equation}
Then Proposition \ref{prop:LiebMattis} may be restated as follows:
$0 = \E_{0}(\mathscr{G}) \leq \E_n(\mathscr{G})$
for all $n\geq 0$,
and there is a strict inequality for all $n>0$ if $\mathscr{G}$ is connected.
\begin{definition}
\label{def:FOEL}
Given $\mathscr{G} = (\mathscr{V},\mathscr{E})$ and $n\in \{0,\dots,\lfloor \frac{1}{2}|\mathscr{V}|\rfloor\}$, we say
$\mathscr{G}$ satisfies FOEL-$n$ (strict FOEL-$n$) if $\E_{n'}(\mathscr{G}) \geq \E_n(\mathscr{G})$ for all $n' \geq n$
($\E_{n'}(\mathscr{G}) > \E_n(\mathscr{G})$ for any $n' > n$).
\end{definition}
By Proposition \ref{prop:LiebMattis} every graph $\mathscr{G}$ satisfies FOEL-$0$,
and it satisfies strict FOEL-$0$  if $\mathscr{G}$ is connected.
%
The FOEL property was introduced in \cite{NachtergaeleSpitzerStarr} as a ferromagnetic analogue
of Lieb and Mattis's ``ordering of energy levels'' property \cite{LiebMattis}, which they proved for antiferromagnets and ferrimagnets on bipartite lattices.


It is of interest to verify FOEL for a  collection of graphs
arising in physical models.
We note that $\Z^d$ has a natural graph structure. 
%
Denoting an element of $\Z^d$ as $\r= (r_1,\dots,r_d)$, the graph distance to the origin
corresponds to the $\ell^1$-norm $\|\r\|_1 = |r_1|+\dots+|r_d|$.
Setting $\mathscr{V} = \Z^d$, then $\{\r,\r'\}$ is an edge in $\mathscr{E}(\Z^d)$
for $\r,\r' \in \Z^d$ if and only if $\|\r-\r'\|_1=1$.
%
Given any subset $\Lambda \subseteq \Z^d$, we will write $\mathscr{G}(\Lambda)$ for the graph
where $\mathscr{V}(\Lambda)=\Lambda$ and $\mathscr{E}(\Lambda)$ is the edge set induced from $\Z^d$:
in other words,
\begin{equation}
\label{eq:InducedEdges}
\mathscr{E}(\Lambda)\, =\, \big\{\{\r,\r'\} \subseteq \Lambda\, :\, \{\r,\r'\} \in \mathscr{E}(\Z^d) \big\}\, .
\end{equation}
When $\Lambda \subset \Z^d$ is finite,
we will write $H_{\Lambda}$ in place of $H_{\mathscr{G}(\Lambda)}$ when no confusion arises from this notation. Similarly, we write $\E_n(\Lambda)$ in place of $\E_n(\mathscr{G}(\Lambda))$.
We will say that $\Lambda$ satisfies (strict) FOEL-$n$ if $\mathscr{G}(\Lambda)$ does.

For each $d$ and $L$ in $\{1,2,\dots\}$ we define the $d$-dimensional box of sidelength $L$
$$
\mathbb{B}^d(L)\, =\,
\big\{\r = (r_1,\dots,r_d) \in \Z^d\, :\, r_1,\dots,r_d \in \{1,\dots,L\}\big\}\, .
$$
This is a subset of $\Z^d$. We will prove the FOEL properties for boxes.
\begin{theorem}[Main Result]
\label{thm:main}
For each choice of $n,d \in \{1,2,\dots\}$, there exists a finite integer $L_{n,d}^0 \in \{1,2,\dots\}$ such that 
$\Lambda = \mathbb{B}^d(L)$ satisfies strict FOEL-$n$ for all $L\geq L_{n,d}^0$.
\end{theorem}

In Figure \ref{fig:OneExample} we have plotted the numerical eigenvalues of $H_{\mathscr{G}}$
against the spin $s = \frac{1}{2}|\mathscr{V}|-n$, for two examples
of graphs: a chain $\{1,\dots,8\}$ (i.e., $B^1(8)$)
viewed as a subgraph of $\Z$, and a $3\times 3$ box (i.e., $B^2(3)$) viewed as a subgraph of $\Z^2$.
The figure shows that they both satisfy FOEL-$n$ for all $n$.
These are both boxes, one of dimension $d=1$ and one of dimension $d=2$.
Our theorem applies to sufficiently large boxes.



\begin{center}
\begin{figure}
\begin{tikzpicture}[xscale=0.9,yscale=0.9]
\begin{scope}[yscale=0.5]
\draw (-1.75,1.5) -- (-1.75,6.75);
\fill (-1.75,1.5) circle (0.75mm and 1.25mm);
\fill (-1.75,2.25) circle (0.75mm and 1.25mm);
\fill (-1.75,3) circle (0.75mm and 1.25mm);
\fill (-1.75,3.75) circle (0.75mm and 1.25mm);
\fill (-1.75,4.5) circle (0.75mm and 1.25mm);
\fill (-1.75,5.25) circle (0.75mm and 1.25mm);
\fill (-1.75,6) circle (0.75mm and 1.25mm);
\fill (-1.75,6.75) circle (0.75mm and 1.25mm);
\draw (-1.75,7.5) node[] {$\mathbb{B}^1(8)$};
\draw[thin] (-0.25,0) -- (4.35,0) node[right] {\small $s$};
\fill (4.15,-0.2) -- (4.35,0) -- (4.15,0.2);
\draw[thin] (4,0.2) -- (4,-0.2) node[below] {\small $4$};
\draw[thin] (3,0.2) -- (3,-0.2) node[below] {\small $3$};
\draw[thin] (2,0.2) -- (2,-0.2) node[below] {\small $2$};
\draw[thin] (1,0.2) -- (1,-0.2) node[below] {\small $1$};
\draw[thin] (0,-0.25) node[below] {\small $0$} -- (0,11) node[left] {\small $E$};
\draw[thin] (0.1,2) -- (-0.1,2) node[left] {\small $1$};
\draw[thin] (0.1,4) -- (-0.1,4) node[left] {\small $2$};
\draw[thin] (0.1,6) -- (-0.1,6) node[left] {\small $3$};
\draw[thin] (0.1,8) -- (-0.1,8) node[left] {\small $4$};
\draw[thin] (0.1,10) -- (-0.1,10) node[left] {\small $5$};
\fill (-0.1,10.75) -- (0,11) -- (0.1,10.75);
\fill (4,0) circle (0.5mm and 1mm);
\foreach \y in {
0.1522, 0.5858, 1.2346, 2.0000, 2.7654, 3.4142}
{
\fill (3,\y) circle (0.5mm and 1mm);
}
\foreach \y in {0.3213, 0.7133, 1.1274, 1.4213, 1.5474, 1.8710, 2.2064, 2.7894, 2.9915, 3.5960, 3.6571, 4.1019, 4.2834, 4.4021, 4.7435, 5.1048, 5.5762, 5.9513, 6.4313, 7.1632}
{
\fill (2,\y) circle (0.5mm and 1mm);
}
\foreach \y in {0.5126, 0.8531, 1.1591, 1.5982, 1.9458, 2.3010, 2.4293, 2.8584, 3.1639, 3.2446, 3.7405, 3.9341, 4.1335, 4.3944, 4.4720, 4.7333, 4.9549, 4.9584, 5.4655, 5.4934, 5.9565, 6.3488, 6.8494, 7.1977, 7.6135, 7.7163, 8.5075, 9.4645}
{
\fill (1,\y) circle (0.5mm and 1mm);
}
\foreach \y in {0.7350, 1.7923, 2.6656, 2.8043, 3.5311, 4.2580, 4.7426, 4.8969, 5.7420, 6.2900, 6.7904, 7.3344, 8.1676, 10.2499}
{
\fill (0,\y) circle (0.5mm and 1mm);
}
\draw[dotted] (4,0) -- (3,0.1522) -- (2,0.3213) -- (1,0.5126) -- (0,0.7350);
\draw[dashed] (4,0) -- (3, 3.4142) -- (2,7.1632) -- (1,9.4645) -- (0,10.2499);
\end{scope}
\begin{scope}[xshift=6.25cm,yscale=0.35]
\draw[thin] (-0.25,0) -- (4.85,0) node[right] {\small $s$};
\fill (4.65,-0.2) -- (4.85,0) -- (4.65,0.2);
\draw[thin] (4.5,0.2) -- (4.5,-0.2) node[below] {\small $4.5$};
\draw[thin] (3.5,0.2) -- (3.5,-0.2) node[below] {\small $3.5$};
\draw[thin] (2.5,0.2) -- (2.5,-0.2) node[below] {\small $2.5$};
\draw[thin] (1.5,0.2) -- (1.5,-0.2) node[below] {\small $1.5$};
\draw[thin] (0.5,0.2) -- (0.5,-0.2) node[below] {\small $0.5$};
\draw[thin] (0,-0.25) -- (0,16) node[left] {\small $E$};
\draw[thin] (0.1,4) -- (-0.1,4) node[left] {\small $2$};
\draw[thin] (0.1,8) -- (-0.1,8) node[left] {\small $4$};
\draw[thin] (0.1,12) -- (-0.1,12) node[left] {\small $6$};
\fill (-0.1,15.75) -- (0,16) -- (0.1,15.75);
\fill (4.5,0) circle (0.5mm and 1mm);
\foreach \y in {    1.0000,    1.0000,    2.0000,   3.0000,   3.0000,    4.0000,    4.0000,    6.0000}
{
\fill (3.5,\y) circle (0.5mm and 1mm);
}
\foreach \y in {     1.6473,    2.0000,    2.3526,    2.5060,    2.5060,    3.0000,
    4.0000,    4.0000,    4.0000,    4.0000,    4.4755,    4.5405,    5.0000,    5.0000,
    5.7205,    5.8901,    5.8901,    6.0000,    6.8912,    7.0000,    7.0000,    7.0000,
    8.0000,    8.1567,    8.6039,    8.6039,   10.2157}
{
\fill (2.5,\y) circle (0.5mm and 1mm);
}
\foreach \y in {    2.3697,    2.3697,    2.5307,    3.5174,    3.6244,    3.6244,
    3.7086,    4.4910,    4.4910,    4.7472,    4.7487,    4.7487,    4.7500,    5.0000,
    5.0000,    5.0000,    5.5171,    6.3116,    6.4157,    6.4157,    6.5117,    6.5117,
    6.6213,    7.0000,    7.0000,    7.0000,    7.2867,    7.2867,    7.4408,    7.9222,
    8.0000,    8.0000,    8.0000,    8.0000,    8.1979,    8.1979,    8.8425,    8.8425,
    9.4152,    9.7291,    9.9586,   11.0492,   11.0492,   11.1421,   11.1710,   11.4626,
   11.4626,   13.5173}
{
\fill (1.5,\y) circle (0.5mm and 1mm);
}
\foreach \y in {    2.8606,    3.5468,    3.5468,    3.7096,    4.4988,    4.5242,
    5.0384,    5.0384,    5.3462,    5.4920,    5.6760,    5.6760,    6.2151,    6.2176,
    6.7055,    6.8322,    6.8322,    7.5405,    7.5751,    7.5751,    7.6743,    8.0000,
    8.0000,    8.2650,    8.2650,    8.7824,    9.0000,    9.3117,    9.3117,    9.5785,
    9.5785,    9.8867,    9.9231,   10.4758,   10.9169,   10.9169,   11.2764,   11.9243,
   12.4482,   13.2594,   13.2594,   15.4987}
{
\fill (0.5,\y) circle (0.5mm and 1mm);
}
\draw[dotted] (4.5,0) -- (3.5,1) -- (2.5,1.6473) -- (1.5,2.3697) -- (0.5,2.8606);
\draw[dashed] (4.5,0) -- (3.5,6.0000) -- (2.5,10.2157) -- (1.5,13.5173) -- (0.5,15.4987);
\draw[thin] (5.75,5) rectangle (7.25,8.5) (6.5,5) -- (6.5,8.5) (5.75,6.75) -- (7.25,6.75);
\fill (5.75,6.75) circle (0.75mm and 1.25mm);
\fill (6.5,6.75) circle (0.75mm and 1.25mm);
\fill (7.25,6.75) circle (0.75mm and 1.25mm);
\fill (5.75,8.5) circle (0.75mm and 1.25mm);
\fill (6.5,8.5) circle (0.75mm and 1.25mm);
\fill (7.25,8.5) circle (0.75mm and 1.25mm);
\fill (5.75,5) circle (0.75mm and 1.25mm);
\fill (6.5,5) circle (0.75mm and 1.25mm);
\fill (7.25,5) circle (0.75mm and 1.25mm);
\draw (6.5,9.75) node[] {$\mathbb{B}^2(3)$};
\end{scope}
\end{tikzpicture}
\begin{caption}{We have plotted the eigenvalues for $H_{\B^d(L)}$ for two graphs.
On the left is $\mathbb{B}^1(8)$. On the right is $\mathbb{B}^2(3)$.
(The graphs are indicated on the sides.)
On the horizontal axes, we have plotted the spin parameter $s=\frac{1}{2} L^d-n$.
In each vertical column above this, we have plotted all the energy eigenvalues \
$E \in \operatorname{spec}(H_{\B^d(L)} \restriction \Hil_{\B^d(L)}^{\mathrm{spin}}(n))$. The lowest energy eigenvalue
in each column is $\E_n(\B^d(L))$. We have plotted a dotted line to aid the eye in comparing these values for different values of $s=\frac{1}{2}L^d-n$.
In both of these graphs $\E_n(\B^d(L))$ is an increasing function of $n$ (decreasing function of $s$). That property is FOEL.
A related property is Lieb and Mattis's antiferromagnetic {\em ordering of energy levels}.
The Lieb-Mattis line is dashed.
\label{fig:OneExample}}
\end{caption}
\end{figure}
\end{center}

\subsection{Outline for the rest of the paper}

In Section \ref{sec:Outline}, we state the two key steps for the proof of Theorem \ref{thm:main}:
an inductive argument for proving FOEL, and an approximate analysis of low energy eigenvectors of the Hamiltonian.
In Section \ref{sec:Proof}, we then give the conditional proof of the main theorem, conditional on the two key steps.
The rest of the paper is dedicated to proving the two key steps.

In Section \ref{sec:Induct}, we prove the first key step, which is an inductive argument for the proof of FOEL.
It relies on some important information for the low energy eigenvectors of the Hamiltonian: not just
one one graph, but 
on an increasing sequence of graphs.
This is what the second key step is devoted to proving.
In Section \ref{sec:KeyII} we prove the second key step.
This is a partial verification of the linear spin wave approximation, at very low energies.
In turn the proof in Section \ref{sec:KeyII} relies on two discrete analogues of Sobolev inequalities,
which are proved in the following two sections.
The first discrete Sobolev inequality is a trace theorem type bound, proved in Section \ref{sec:Trace}.
The second discrete Sobolev inequality is an extension theorem, which is proved in \ref{sec:Extension}.
A number of elementary results are relegated to appendices, which appear at the end of the paper.

{\small
\tableofcontents
}



%
%


\section{The two key steps of the proof}
\label{sec:Outline}


The proof of Theorem \ref{thm:main} relies on two main ideas.
The first key idea is an inductive argument for FOEL.
To obtain the induction step for that argument, we need some asymptotic approximation
to the low energy spectrum. 
The second main idea is a method to obtain that asymptotic approximation.
It is a mathematical verification of part of the physicists' {\em linear spin wave approximation}.
This follows by using some discrete versions of Sobolev type inequalities
to show that states with low energy are {\em spread out},
which seems to be a key physical assumption.

\subsection{Step I: An Inductive Argument for Establishing FOEL}
\label{subsec:Induct}

The main idea for an inductive proof of FOEL is the following.
\begin{proposition}
\label{prop:inductive1}
Suppose that $\mathscr{G} = (\mathscr{V},\mathscr{E})$ is a finite graph, and that
$\mathscr{G}' = (\mathscr{V}',\mathscr{E}')$ is a graph such that $\mathscr{V}' = \mathscr{V} \cup \{x'\}$ for a 
single vertex $x'$ not in $\mathscr{V}$, and such that $\mathscr{E} \subseteq \mathscr{E}'$. 
Suppose 
that $\mathscr{G}$ satisfies FOEL-$n$,
for some $n \in \{0,1,\dots,\lfloor \frac{1}{2}|\mathscr{V}| \rfloor\}$.
Suppose, further that the following condition is satisfied
$$
\E_{n}(\mathscr{G}')\, \leq\, \E_n(\mathscr{G})\, .
$$
Then $\mathscr{G}'$ also satisfies FOEL-$n$.
\end{proposition}
This proposition was originally proved in \cite{KomaNachtergaele} for the special case of $n=1$.
In \cite{NachtergaeleSpitzerStarr}, the present authors gave the argument generalizing the proof
to higher $n$.
In Section \ref{sec:Induct} we will  state and prove a generalization of this fact.
The argument from \cite{KomaNachtergaele} was also, independently, rediscovered in \cite{HandjaniJungreis},
slightly later.
\begin{corollary}
\label{cor:inductive1}
For any $n \in \N$ and for any $N \in \{n+1,n+2,\dots\}$, suppose that $\mathscr{G}_{2n},\mathscr{G}_{2n+1},\dots,\mathscr{G}_{N}$ forms a family of graphs, 
satisfying the following conditions.
\begin{itemize}
\item[(i)] For each $k \in \{2n,\dots,N\}$, we have that $\mathscr{G}_k = (\mathscr{V}_k,\mathscr{E}_k)$ satisfies $|\mathscr{V}_k|=k$.
\item[(ii)] For each $k \in \{2n,\dots,N-1\}$, we have that 
\begin{itemize}
\item[(a)] $\mathscr{V}_{k} \subset \mathscr{V}_{k+1}$, 
\item[(b)] $\mathscr{E}_k \subset \mathscr{E}_{k+1}$,
\item[(c)] $|\mathscr{V}_{k+1} \setminus \mathscr{V}_k| = 1$.
\end{itemize}
\item[(iii)] For each $k \in \{2n,\dots,N-1\}$, we have that $\E_n(\mathscr{G}_{k+1}) \leq \E_n(\mathscr{G}_{k})$. 
\end{itemize}
Then $\mathscr{G}_k$ satisfies FOEL-$n$ for each $k \in \{2n,\dots,N\}$.
\end{corollary}
\begin{proof}
According to Definition \ref{def:FOEL}, the graph $\mathscr{G}_{2n}$ trivially satisfies FOEL-$n$, since $\frac{1}{2}|\mathscr{V}_n|  = n$ so that the only choice for $n'$ greater than or equal to $n$ (and less than or equal to $\frac{1}{2}|\mathscr{V}_n|$) is 
$n$, itself.
But then, using Proposition \ref{prop:inductive1}, we may establish that $\mathscr{G}_k$ satisfies FOEL-$n$ for all $k \in \{2n,\dots,N\}$, with $k=2n$ just established as the initial step.
\end{proof}
This corollary was used in \cite{NachtergaeleSpitzerStarr,NachtergaeleStarr} and \cite{NachtergaeleNgStarr} to prove FOEL-$n$ for all $n$
in  one-dimensional models.
Let us say that graphs $\mathscr{G}_{2n},\mathscr{G}_{2n+1},\dots,\mathscr{G}_{N}$ satisfying conditions (i) and (ii) from the hypotheses of Corollary \ref{cor:inductive1}
are a ``growing family of graphs.''
\begin{proposition}
\label{prop:inductiveNewLow}
Suppose $n \in \N$ and $N \in \{n+1,n+2,\dots\}$ are fixed. Suppose  $\mathscr{G}_{2n},\mathscr{G}_{2n+1},\dots,\mathscr{G}_{N}$ are a growing family
of graphs.
In place of condition (iii), suppose that the following holds:
\begin{itemize}
\item[(iii')] $\min\{\E_n(\mathscr{G}_k)\, :\, k \in \{2n,\dots,N\}\} = \E_n(\mathscr{G}_N)$.
\end{itemize}
Then
\begin{equation}
\label{eq:NewLowConclusion}
\E_n(\mathscr{G}_N)\, \leq\, \min\{\E_r(\mathscr{G}_k)\, :\, k \in \{2n,\dots,N\}\, ,\ r \in \{n,\dots,\lfloor \textstyle{\frac{1}{2}}|\mathscr{V}_k| \rfloor\}\, .
\end{equation}
In particular, $\mathscr{G}_N$ satisfies FOEL-$n$.
\end{proposition}
We will prove this result in Section \ref{sec:Induct}.
For the special case of $n=1$, the analogous result was proved in \cite{ConomosStarr}.
It was also, independently, and slightly later, proved in \cite{Morris} for the special case of $n=1$.
The present extension to $n>1$, follows immediately by the same argument.
A key part of the idea
is to introduce coupling constants (known as ``rates'' in the probability and stochastic processes community)
for the purpose of diluting the graphs.
The idea of including rates in such arguments can be traced back at least to Handjani and Jungreis \cite{HandjaniJungreis}.

The hypothesis of Proposition \ref{prop:inductiveNewLow} is weaker than that of Corollary \ref{cor:inductive1} because the condition (iii')
is weaker than the condition (iii).
But the conclusion is also weaker since we only conclude FOEL-$n$ for the final graph $\mathscr{G}_N$,
as opposed to all graphs $\mathscr{G}_k$, for $k \in \{2n,\dots,N\}$.
%
For a certain family of graphs, verifying condition (iii') is easier than
trying to verify condition (iii).
Let us define the family of graphs, now.
\begin{definition}
\label{def:LambdadN}
Given $d \in \N$ and for each $N \in \N$, define $L(d,N) = \lfloor N^{1/d} \rfloor$ and $L^+(d,N) = \lceil N^{1/d} \rceil$.
Let $\Lambda(d,N)\subset\Z^d$ be defined as follows. If $L^+(d,N) = L(d,N)$,
then $\Lambda(d,N) \stackrel{\mathrm{def}}{:=} \B^d_{L(d,N)}$.
Otherwise, 
let $\prec_d$ denote the lexicographic ordering on $\Z^d$, and
let $\Lambda(d,N) \stackrel{\mathrm{def}}{:=} \B^d_{L(d,N)} \cup S(d,N)$ where $S(d,N)$ is the set of cardinality $N-(L(d,N))^d$, consisting of points in $\B^d_{L^+(d,N)}\setminus\B^d_{L(d,N)}$
which are smallest with respect to $\prec_d$.
\end{definition}

\subsection{Step II: Analysis of low energy wave-functions}
\label{subsec:Magnon}

We continue to consider a general finite graph $\mathscr{G} = (\mathscr{V},\mathscr{E})$.
For any $n \in \{0,\dots,|\mathscr{V}|\}$, we may consider the action of $H_{\mathscr{G}}$,
restricted to $\Hil_{\mathscr{V}}^{\mathrm{mag}}(n)$,
in some coordinates.

From  definition (\ref{eq:HamDef}) we know $H_{\mathscr{G}}$ is a self-adjoint operator on $\Hil_{\mathscr{V}}$.
For any $E_1,E_2 \in \R$ with $E_1\leq E_2$, let $\mathfrak{L}_{\mathscr{G}}(E_1,E_2)$
denote the subspace
\begin{equation}
\label{eq:espace}
\mathfrak{L}_{\mathscr{G}}(E_1,E_2)\,
=\, \bigoplus_{\lambda \in \operatorname{spec}(H_{\mathscr{G}}) \cap [E_1,E_2]}
\ker(H_{\mathscr{G}} - \lambda \mathbbm{1})\, .
\end{equation}
Note that by Proposition \ref{prop:LiebMattis},  if $E_1\leq 0\leq E_2$,
then $\mathfrak{L}_{\mathscr{G}}(E_1,E_2) 
=\mathfrak{L}_{\mathscr{G}}(0,E_2)$.

Now we state the approximate eigenvectors.
\begin{definition}
\label{eq:defTn0}
Recalling the definition in
(\ref{eq:OmegaDEFn}),
we define $T_{\mathscr{V}}^{(n)} : \ell^2(\mathscr{V}^n) \to \Hil_{\mathscr{V}}$
by
$$
T_{\mathscr{V}}^{(n)} F\, =\, \frac{1}{\sqrt{n!}}\, \sum_{(x_1,\dots,x_n) \in \mathscr{V}^n} F(x_1,\dots,x_n) \Omega_{\mathscr{V}}^{(n)}(x_1,\dots,x_n)\, .
$$
\end{definition}



\begin{definition}
\label{def:littleF}
For each $\xi \in \R$ define the function $f(\xi,\cdot) : \Z \to \R$ as follows: for $\xi \in \R \setminus \{0\}$,
define
$f(\xi,r) = 2^{1/2} \cos(\pi\xi [r-\frac{1}{2}])$ for all $r \in \Z$;
and $f(0,r) = 1$ for all $r \in \Z$.
\end{definition}
We denote an element of $\{0,1,\dots\}^d$ as
$\k = (\kappa_1,\dots,\kappa_d)$.
Sometimes we need several vectors, such as $\k_1,\dots,\k_n \in \{0,1,\dots\}^d$.
In this case, the coordinates will be expressed as $\k_k = (\k_{k,1},\dots,\k_{k,d})$ for each $k \in \{1,\dots,n\}$
(and similarly for similar cases).
\begin{definition}
\label{def:PsiTilde}
Fix $d,n \in \N$.
For $\k_1,\dots,\k_n \in \{0,1,\dots\}^d$,
and each $N \in \N$,
we define a function, $F_{d,N}^{(n)}(\k_1,\dots,\k_n;\cdot) \in \ell^2(\Lambda(d,N)^n)$, by the formula
\begin{equation}
\label{eq:FdNDef}
F_{d,N}^{(n)}(\k_1,\dots,\k_n;\r_1,\dots,\r_n)
=\, [L(d,N)]^{-nd/2} \sum_{\pi \in S_n} \prod_{k=1}^{n} \prod_{j=1}^{d} f\big([L^+(d,N)]^{-1} \kappa_{k,j},r_{\pi_k,j}\big)\, ,
\end{equation}
which is in the range of $\mathfrak{S}_{\Lambda(d,N)}^{(n)}$ because of the averaging over the action of $S_n$.
Then we define a  vector, 
$\widetilde{\Psi}^{(n)}_{d,N}(\k_1,\dots,\k_n) \in \Hil_{\Lambda(d,N)}^{\mathrm{mag}}(n)$,
by the formula
$\widetilde{\Psi}^{(n)}_{d,N}(\k_1,\dots,\k_n)
= T_{\Lambda(d,N)}^{(n)} F_{d,N}^{(n)}(\k_1,\dots,\k_n;\cdot)$.
\end{definition}
We put the tilde over the vector because these are not actual eigenvectors.
But we will see that they are ``approximate eigenvectors.''
%
%
%
One immediate corollary of the definition is this:
\begin{lemma}
Fix $d,n \in \N$. For any choice of $\k_1,\dots,\k_n \in \{0,1,\dots\}^d$,
\begin{equation}
\label{eq:SpinLoweringApproximate}
S_{\Lambda(d,N)}^- \widetilde{\Psi}^{(n)}_{d,N}(\k_1,\dots,\k_n)\,
=\, 
[L(d,N)]^{d/2}\, \widetilde{\Psi}^{(n+1)}_{d,N}(\k_1,\dots,\k_n,\boldsymbol{0})\, .
\end{equation}
\end{lemma}
\begin{proof}
For a general graph $\mathscr{G} = (\mathscr{V},\mathscr{E})$, we have
$T_{\mathscr{V}}^{(n+1)} \widehat{S}_{\mathscr{V}}^{(n,n+1)}=S_{\mathscr{V}}^- T_{\mathscr{V}}^{(n)}$, 
where $\widehat{S}_{\mathscr{V}}^{(n,n+1)} : \ell^2(\mathscr{V}^n) \to \ell^2(\mathscr{V}^{n+1})$ is the linear transformation such that
$$
\widehat{S}_{\mathscr{V}}^{(n,n+1)} F(x_1,\dots,x_{n+1})\, =\,  F(x_1,\dots,x_n) + F(x_2,\dots,x_{n+1}) + \sum_{k=2}^{n} F(x_1,\dots,x_{k-1},x_{k+1},\dots,x_{n+1})\, .
$$
But then, from (\ref{eq:FdNDef}), we see that this is equivalent to (\ref{eq:SpinLoweringApproximate}).
\end{proof}

For any $\pi \in S_n$, we have 
\begin{equation}
\label{eq:PsiSymm}
\widetilde{\Psi}_{d,N}^{(n)}(\k_{\pi_1},\dots,\k_{\pi_n})\,
=\, \widetilde{\Psi}_{d,N}^{(n)}(\k_1,\dots,\k_n)\, ,
\end{equation}
which we see directly by inspection of Definition \ref{eq:defTn0} and Definition \ref{def:PsiTilde}.
For $d,n \in \N$, let 
$$
\mathcal{K}(d,n)\, \stackrel{\mathrm{def}}{:=}\, 
(\{0,1,\dots\}^d)^n\, 
=\, \{(\k_1,\dots,\k_n)\, :\, \k_1,\dots,\k_n \in \{0,1,\dots\}^d\}\, .
$$
Given any $(\k_1,\dots,\k_n) \in \mathcal{K}(d,n)$, define
$$
[(\k_1,\dots,\k_n)]_{S_n} = \{(\k_{\pi_1},\dots,\k_{\pi_n})\, :\, \pi \in S_n\} \subset \mathcal{K}(d,n)\, .
$$
Let us define 
\begin{equation}
\label{eq:widetildeK}
\widetilde{\mathcal{K}}(d,n)\, =\, \{[(\k_1,\dots,\k_n)]_{S_n}\,
:\, (\k_1,\dots,\k_n) \in \mathcal{K}(d,n)\}\, ,
\end{equation} 
which is the quotient space $\mathcal{K}(d,n)/\sim$,
where the equivalence relation is: $(\k_1,\dots,\k_n) \sim (\k_1',\dots,\k_n')$
if and only if $[(\k_1,\dots,\k_n)]_{S_n} = [(\k_1',\dots,\k_n')]_{S_n}$.

An easy variational calculation follows.
Let us define a quantity (related to the spectral gap)
\begin{equation}
\label{eq:gammaDef}
\gamma\, \stackrel{\mathrm{def}}{:=}\, \pi^2/2\, .
\end{equation}
Then we have the following:
\begin{proposition}
\label{prop:Variational1}
Fix $d,n \in \N$. For any choice of $\k_1,\dots,\k_n \in \{0,1,\dots\}^d$,
\begin{equation}
\label{eq:AsymptoticEnergy}
\lim_{N \to \infty} \bigg\| \bigg(\gamma^{-1}[L(d,N)]^2 H_{\Lambda(d,N)} - \sum_{k=1}^{n} \sum_{j=1}^{d} \kappa_{k,j}^2 \bigg)
\widetilde{\Psi}^{(n)}_{d,N}(\k_1,\dots,\k_n)\bigg\|\,
=\, 0\, .
\end{equation}
Moreover, for any additional choice of $\k_1',\dots,\k_n' \in \{0,1,\dots\}^d$,
\begin{equation}
\label{eq:AsymptoticIP}
\lim_{N \to \infty} \big\langle \widetilde{\Psi}_{d,N}^{(n)}(\k_1,\dots,\k_n),
\widetilde{\Psi}_{d,N}^{(n)}(\k_1',\dots,\k_n') \big\rangle\,
=\, 
n! \sum_{\pi \in S_n} \prod_{k=1}^{n} \mathbf{1}\{\k_{\pi_k}=\k'_{k}\}\, .
\end{equation}
\end{proposition}
We will prove this in Appendix \ref{app:AsymptoticEnergy}.
It is easy, using the explicit formulas
for the eigenvectors.

From (\ref{eq:PsiSymm}) and (\ref{eq:AsymptoticIP}), we see that the appropriate
labeling of approximate eigenstates is given by elements of $\widetilde{\mathcal{K}}(d,n)$
from (\ref{eq:widetildeK}).
But from (\ref{eq:AsymptoticEnergy}), we also see that the appropriate energy
for $\widetilde{\Psi}^{(n)}_{d,N}(\k_1,\dots,\k_n)$, modulo the scale $\gamma\cdot [L(d,N)]^{-2}$,
is $\sum_{k=1}^{n} \sum_{j=1}^{d} \kappa_{k,j}^2$ for $(\k_1,\dots,\k_n) \in \mathcal{K}(d,n)$.
For each $m \in \{0,1,\dots\}$, let us define
\begin{equation}
\mathcal{K}(d,n,m)\, \stackrel{\mathrm{def}}{:=}\, \bigg\{ (\k_1,\dots,\k_n) \in \mathcal{K}(d,n)\, :\,
\sum_{k=1}^{n} \sum_{j=1}^{d} \kappa_{k,j}^2=m\bigg\}\, .
\end{equation}
Let us also define, for each $m \in \{0,1,\dots\}$, 
\begin{equation}
\label{eq:tildecalKdnmDef}
\widetilde{\mathcal{K}}(d,n,m)\,
=\,
\{[(\k_1,\dots,\k_n)]_{S_n}\,
:\, (\k_1,\dots,\k_n) \in \mathcal{K}(d,n,m)\}\, .
\end{equation}
Finally, let us define
the number $R(d,n,m) \in \{0,1,\dots\}$,
\begin{equation}
R(d,n,m)\, =\, 
|\widetilde{\mathcal{K}}(d,n,m)|\, ,
\end{equation}
which is supposed to enumerate the linearly independent (indeed orthogonal) eigenvectors whose 
eigenvalues, modulo the scale $\gamma\cdot [L(d,N)]^{-2}$,
are all near to $m$.
Now we state the main theorem for the linear spin wave approximation that we will use.
\begin{theorem}
\label{thm:SpecGather}
Fix $d,n \in \N$. For each $N \in \N$, we may define vectors $\Psi^{(n)}_{d,N}(m,r)$,
for each $m \in \{0,1,\dots\}$ and each $r\in \{1,\dots,R(d,n,m)\}$,
which are all orthogonal, and such that,
for each $m \in \{0,1,\dots\}$ and each $\epsilon \in (0,1/2)$, there is an integer $N_1(d,n,m,\epsilon) \in \N$ such that
for all $N\geq N_1(d,n,m,\epsilon)$ the following properties hold.
\begin{itemize}
\item[(i)] For each $m_1 \in \{0,\dots,m\}$ and each $r\in\{1,\dots,R(d,n,m_1)\}$, the vector $\Psi^{(n)}_{d,N}(m_1,r)$
is an eigenvector of $\gamma^{-1} [L(d,N)]^2 H_{\Lambda(d,N)}$ with associated eigenvalue $E^{(n)}_{d,N}(m_1,r) \in (m_1-\epsilon,m_1+\epsilon)$.
\item[(ii)]
The vectors $\{\Psi^{(n)}_{d,N}(m_1,r)\, :\, m_1 \in \{0,\dots,m\}\, ,\ r\in\{1,\dots,R(d,n,m_1)\}\}$ 
are an orthonormal basis for $\mathfrak{L}_{\Lambda(d,N)}(0,\gamma^{-1} [L(d,N)]^2 (m+\epsilon)) \cap \Hil^{\mathrm{mag}}_{\Lambda(d,N)}(n)$, 
using the definition from (\ref{eq:espace}).
\end{itemize}
\end{theorem}

\section{Conditional Proof of the Main Theorem}
\label{sec:Proof}

The following easy result is the first step to a calculation technique.
\begin{lemma}
\label{lem:specMon}
For any finite graph $\mathscr{G} = (\mathscr{V},\mathscr{E})$ and any $n \in \{0,\dots,\lfloor \frac{1}{2}|\mathscr{V}|\rfloor\}$:\\
(1) for each choice of $E_1,E_2 \in \R$ with $0\leq E_1\leq E_2$, we have
$$
\dim(\mathfrak{L}_{\mathscr{G}}(E_1,E_2) \cap \Hil^{\mathrm{mag}}_{\mathscr{V}}(n))\, \geq\,
\dim(\mathfrak{L}_{\mathscr{G}}(E_1,E_2) \cap \Hil^{\mathrm{mag}}_{\mathscr{V}}(n-1))\, ;
$$
(2) moreover, we may calculate $\mathfrak{E}_n(\mathscr{G})$ from (\ref{eq:EDef}) as
$$
\mathfrak{E}_n(\mathscr{G})\, =\, \min(\{E \in [0,\infty)\, :\, \dim(\mathfrak{L}_{\mathscr{G}}(0,E) \cap \Hil^{\mathrm{mag}}_{\mathscr{V}}(n))>
\dim(\mathfrak{L}_{\mathscr{G}}(0,E) \cap \Hil^{\mathrm{mag}}_{\mathscr{V}}(n-1))\})\, .
$$
\end{lemma}
\begin{proof}
By Lemma \ref{lem:preliminary}, we know that $S_{\mathscr{V}}^-$ is an isomorphism of $\Hil_{\mathscr{V}}(m,n)$ onto $\Hil_{\mathscr{V}}(m+1,n)$
for each $m \in \{n,\dots,|\mathscr{V}|-n-1\}$. Since this operator commutes with $H_{\mathscr{G}}$, we may simplify the definition of $\mathfrak{E}_n(\mathscr{G})$
from (\ref{eq:EDef}):
$$
\E_n(\mathscr{G})\, =\, \min\operatorname{spec}\big(H_{\mathscr{G}}
\restriction \Hil_{\mathscr{V}}(n,n)\big)\, .
$$
But, then we see that, again by Lemma \ref{lem:preliminary}, 
$$
\Hil_{\mathscr{V}}(n,n)\, =\, \Hil_{\mathscr{V}}^{\mathrm{mag}}(n) \cap \Big( S_{\mathscr{V}}^-\big(\Hil_{\mathscr{V}}^{\mathrm{mag}}(n-1)\big) \Big)^{\perp}\, .
$$
Also, by Lemma \ref{lem:preliminary}, we know that $S_{\mathscr{V}}^-$ is an isomorphism of $\Hil_{\mathscr{V}}^{\mathrm{mag}}(n-1)$ onto 
$S_{\mathscr{V}}^-\big(\Hil_{\mathscr{V}}^{\mathrm{mag}}(n-1)\big)$.
So, again, since $S_{\mathscr{V}}^-$ commutes with $H_{\mathscr{G}}$, we may conclude both facts stated in the lemma.
\end{proof}
%
\begin{corollary}
\label{cor:EnAsympFrakm}
Fix $d,n \in \{1,2,\dots\}$.
Define $\mathfrak{m}(d,n) \in \{1,2,\dots\}$ as 
$$
\mathfrak{m}(d,n)\, \stackrel{\mathrm{def}}{:=}\, \min(\{m \in \{1,2,\dots\}\, :\, R(d,n,m)>R(d,n-1,m)\})\, .
$$
For any $\epsilon>0$, and $N\geq N_1(d,n,\mathfrak{m}(d,n),\epsilon)$ (as in Theorem \ref{thm:SpecGather}), we have
$$
\left|\gamma^{-1} [L(d,N)]^2 \mathfrak{E}_n(\Lambda(d,N)) - \mathfrak{m}(d,n)\right|\, \leq\, \epsilon\, .
$$
\end{corollary}
Before proving this simple corollary of Theorem \ref{thm:SpecGather} and Lemma \ref{lem:specMon}, let us note this formula:
\begin{lemma}
\label{lem:mathfrakm}
For each $d,n \in \{1,2,\dots\}$, we have
\begin{equation}
\label{eq:mfrakForm}
\mathfrak{m}(d,n)\, =\, n\, .
\end{equation}
\end{lemma}
\begin{proof}
For any $m<n$, choose an element of $\widetilde{\mathcal{K}}(d,n,m)$.
We may find some $(\k_1,\dots,\k_{n}) \in \mathcal{K}(d,n,m)$ such that the chosen element is $[(\k_1,\dots,\k_{n})]_{S_n}$.
But then $\|\k_1\|^2+\dots+\|\k_n\|^2=m$. The minimum nonzero value for $\|\k_j\|^2$ is $1$
for each $j$.
Thus, $m<n$.
So $\k_j=\boldsymbol{0}$ for some $j \in \{1,\dots,n\}$.
Without loss of generality (due to the permutation symmetry), we may assume that the $(\k_1,\dots,\k_{n})$
that we found had $\k_n=\boldsymbol{0}$.
This means that $[(\k_1,\dots,\k_{n-1})]_{S_{n-1}}$
is an element of $\widetilde{\mathcal{K}}(d,n-1,m)$.
Since we have merely removed one zero element, and we mod out by the action of the symmetric
group in the quotient space, this mapping is a bijection. 
So $R(d,n-1,m)=R(d,n,m)$ for each $m<n$.
More generally, even for $m\geq n$, there is always a bijection between the set of elements 
of $\widetilde{\mathcal{K}}(d,n,m)$ having at least one $\boldsymbol{0}$ part
and the set of elements of $\widetilde{\mathcal{K}}(d,n-1,m)$
But, on the other hand, for 
$\boldsymbol{\delta}_{d,1} \stackrel{\mathrm{def}}{:=} (1,0,\dots,0) \in \{0,1,\dots\}^d$, 
we may see that
$[(\boldsymbol{\delta}_{d,1},\dots,\boldsymbol{\delta}_{d,1})]_{S_n}$ is an element
of $\widetilde{\mathcal{K}}(d,n,n)$  not equal to $[(\k_1,\dots,\k_{n-1},\boldsymbol{0})]_{S_n}$
for any choice of  $[(\k_1,\dots,\k_{n-1})]_{S_{n-1}}$
in $\widetilde{\mathcal{K}}(d,n-1,m)$.
So $\widetilde{\mathcal{K}}(d,n,n)$ is not bijective to $\widetilde{\mathcal{K}}(d,n-1,m)$.
Rather, $R(d,n,n) > R(d,n-1,n)$.
\end{proof}
\begin{proofof}{\bf Proof of Corollary \ref{cor:EnAsympFrakm}:}
Let $\widehat{\mathfrak{L}}_{\Lambda(d,N)}(E_1,E_2) = \mathfrak{L}_{\Lambda(d,N)}(\gamma\cdot [L(d,N)]^{-2}E_1,\gamma\cdot [L(d,N)]^{-2}E_2)$.
Let $m = \mathfrak{m}(d,n)$ (which is $n$).
We know, by Theorem \ref{thm:SpecGather}, that for any $m_1 \in \{1,2,\dots\}$ with $m_1<m$, we have
$$
\dim\left(\widehat{\mathfrak{L}}_{\Lambda(d,N)}(0,m_1+\epsilon) \cap \Hil^{\mathrm{mag}}_{\Lambda(d,N)}(n)\right)\, =\, \sum_{m_2=0}^{m_1} R(d,n,m_2)\, .
$$
But since $m_1<m$, we have that this is also equal to $\sum_{m_2=0}^{m_1} R(d,n-1,m_2)$.
Therefore, we have
$$
\dim\left(\widehat{\mathfrak{L}}_{\Lambda(d,N)}(0,m_1+\epsilon) \cap \Hil^{\mathrm{mag}}_{\Lambda(d,N)}(n)\right)\, =\, 
\dim\left(\widehat{\mathfrak{L}}_{\Lambda(d,N)}(0,m_1+\epsilon) \cap \Hil^{\mathrm{mag}}_{\Lambda(d,N)}(n-1)\right)\, .
$$
So $\gamma^{-1}\cdot[L(d,N)]^2\mathfrak{E}_n(\Lambda(d,N))>m_1+\epsilon$, by Lemma \ref{lem:specMon}.
By Theorem \ref{thm:SpecGather}, we also know that $\dim(\widehat{\mathfrak{L}}_{\Lambda(d,N)}(m-1+\epsilon,m-\epsilon) \cap \Hil^{\mathrm{mag}}_{\Lambda(d,N)}(n))$ equals $0$.
Hence, $\gamma^{-1}\cdot[L(d,N)]^2\mathfrak{E}_n(\Lambda(d,N))\geq m-\epsilon$.
But we also know, by Theorem \ref{thm:SpecGather},
\begin{align*}
\dim\left(\widehat{\mathfrak{L}}_{\Lambda(d,N)}(0,m+\epsilon) \cap \Hil^{\mathrm{mag}}_{\Lambda(d,N)}(n)\right)\, 
&=\, 
\sum_{m_1=0}^{m-1} R(d,n,m_1) + R(d,n,m)\\
&=\, \sum_{m_1=0}^{m-1} R(d,n-1,m_1) + R(d,n,m)\\
&>\, \sum_{m_1=0}^{m-1} R(d,n,m_1) + R(d,n-1,m)\\
&=\, \dim\left(\widehat{\mathfrak{L}}_{\Lambda(d,N)}(0,m+\epsilon) \cap \Hil^{\mathrm{mag}}_{\Lambda(d,N)}(n-1)\right)\, .
\end{align*}
So $\gamma^{-1} [L(d,N)]^2\mathfrak{E}_n(\Lambda(d,N))\leq m+\epsilon$.
\end{proofof}

Combining Corollary \ref{cor:EnAsympFrakm} and  Lemma \ref{lem:mathfrakm}, we conclude the following.
\begin{corollary}
\label{cor:EnAsymptoticsn}
Fix $d,n \in \{1,2,\dots\}$. Then, 
$$
\mathfrak{E}_n(\Lambda(d,N))\, \sim\, \gamma \cdot [L(d,N)]^{-2} \cdot n\, ,\quad \text{ as $N \to \infty$.}
$$
\end{corollary}
\begin{proof}
Corollary \ref{cor:EnAsympFrakm} and equation (\ref{eq:mfrakForm}) from  Lemma \ref{lem:mathfrakm} imply that, for each $\epsilon>0$, we have $\mathfrak{E}_n(\Lambda(d,N))/(\gamma \cdot [L(d,N)]^{-2} \cdot n)$
is in the interval $[1-\epsilon,1+\epsilon]$ for sufficiently large $N$:
specifically,
$N>N_1(d,n,n,\epsilon)$, from Theorem \ref{thm:SpecGather}. That is the definition of asymptotic equivalence.
\end{proof}

Recall from Definition \ref{def:LambdadN} that $L(d,N) = \lfloor N^{1/d} \rfloor$.
In particular, this means $L(d,N) \sim N^{1/d}$, as $N \to \infty$.
So, Corollary \ref{cor:EnAsymptoticsn} may be rewritten as
\begin{equation}
\label{eq:EnAsymptoticAlgebraic}
\E_n(\Lambda(d,N))\, \sim\, n \gamma N^{-2/d}\, , \quad \text{ as $N \to \infty$.}
\end{equation}
We will use (\ref{eq:EnAsymptoticAlgebraic}).
We also need the following:
\begin{lemma}
\label{lem:realAn}
Suppose that $(t_{a},t_{a+1},\dots)$ is a sequence of strictly positive numbers such that $t_N \sim C N^{-p}$ for some $C \in (0,\infty)$ and some $p>0$.
Let us define a function $\nu : \{a,a+1,\dots\} \to \{a,a+1,\dots\} \cup \{\infty\}$ by the formula
$$
\nu(N)\,
=\, \inf\{N' \in \{N,N+1,\dots\}\, :\, t_{N'} = \min\{t_a,\dots,t_{N'}\}\}\, .
$$
Then $t_{\nu(N)} \sim C N^{-p}$.
\end{lemma}
This is an exercise in advanced calculus. For completeness we include its proof in Appendix \ref{app:Silly}.
%
Now we may give the conditional proof of our main theorem.

\begin{proofof}{\bf Proof of Theorem \ref{thm:main}:}
This follows from Corollary \ref{cor:EnAsymptoticsn}, as rewritten in equation (\ref{eq:EnAsymptoticAlgebraic}), as well as Lemma \ref{lem:realAn},
and Proposition \ref{prop:inductiveNewLow}.
For each $n$, let $\nu_n$ be defined as 
\begin{equation}
\label{eq:def-nu-n}
\nu_n(N)\, =\, \inf\{N' \in \{N,N+1,\dots\}\, :\, \E_n(\Lambda(d,N')) = \min\{\E_n(\Lambda(d,2n)),\dots,\E_n(\Lambda(d,N'))\}\}\, .
\end{equation}
Then we know that $\E_n(\Lambda(d,\nu_n(N))) \sim C_n N^{-p}$ by (\ref{eq:EnAsymptoticAlgebraic}), where $C_n = n\gamma$ and $p=2/d$.
The hypotheses of Proposition \ref{prop:inductiveNewLow}
are satisfied for $n$ and the finite sequence of graphs 
$$
(\Lambda(d,2n),\dots,\Lambda(d,\nu_n(N))
$$ 
by (\ref{eq:def-nu-n}).
So we know that
$$
\min_{N' \in \{2n,\dots,\nu_n(N)\}}\, \min_{n' \in \{n,\dots,\lfloor N'/2 \rfloor\}} \E_{n'}(\Lambda(d,N'))\, =\, \E_n(\Lambda(d,\nu_n(N)))\, .
$$
In particular, since $\nu_n(N)\geq N$, we have
$$
\min_{n' \in \{n,\dots,\lfloor N/2 \rfloor\}} \E_{n'}(\Lambda(d,N))\, \geq\, \E_n(\Lambda(d,\nu_n(N)))\, .
$$
But as we already established, $\E_n(\Lambda(d,\nu_n(N))) \sim C_n N^{-p}$. Therefore, defining a new quantity $\E^{\min}_{n\uparrow}(\Lambda(d,N))$ as
\begin{equation}
\label{eq:Def-E-min-n-plus}
\E^{\min}_{n\uparrow}(\Lambda(d,N))\, =\, \min_{n' \in \{n,\dots,\lfloor N/2 \rfloor\}} \E_{n'}(\Lambda(d,N))\, ,
\end{equation}
we actually have $\liminf_{N \to \infty} \E^{\min}_{n\uparrow}(\Lambda(d,N)) / (C_n N^{-p}) \geq 1$.
But also, $\E^{\min}_{n\uparrow}(\Lambda(d,N)) \leq \E_n(\Lambda(d,N))$ because $\E_n(\Lambda(d,N))$
is one of the terms in the minimum on the right hand side (\ref{eq:Def-E-min-n-plus}). So since $\E_n(\Lambda(d,N)) \sim C_n N^{-p}$, as $N \to \infty$, we also have
$$
\limsup_{N \to \infty} \frac{\E^{\min}_{n\uparrow}(\Lambda(d,N))}{C_n N^{-p}}\, \leq\, 1\, .
$$
In other words, putting these two inequalities together,
\begin{equation*}
\E^{\min}_{n\uparrow}(\Lambda(d,N))\, \sim\, C_n N^{-p}\, ,\, \text{ as $N \to \infty$.}
\end{equation*}
But now this allows us to get the result we want, because this is true for all $n$.
In particular, fixing $n$, it is also true for $n+1$. So
\begin{equation*}
\E^{\min}_{n+1\uparrow}(\Lambda(d,N))\, \sim\, C_{n+1} N^{-p}\, .
\end{equation*}
But $C_{n+1} = (n+1)\gamma$ is strictly greater than $C_n=n\gamma$.
From the definition  (\ref{eq:Def-E-min-n-plus}) again,
$$
\E^{\min}_{n\uparrow}(\Lambda(d,N))\, =\, \min\{\E_n(\Lambda(d,N)),\E^{\min}_{n+1\uparrow}(\Lambda(d,N))\}\, .
$$
But $\E_n(\Lambda(d,N)) \sim C_n N^{-p}$ and $\E^{\min}_{n+1\uparrow}(\Lambda(d,N)) \sim C_{n+1} N^{-p}$.
The strict inequality $C_n < C_{n+1}$ implies that the
former is strictly smaller than the latter for sufficiently large $N$.
So
$$
\E^{\min}_{n\uparrow}(\Lambda(d,N))\, =\, \E_n(\Lambda(d,N))\, ,\text{ and }\
\E^{\min}_{n\uparrow}(\Lambda(d,N))\, <\, \E^{\min}_{n+1\uparrow}(\Lambda(d,N))\, ,
$$
for sufficiently large $N$.
Decoding the definition  (\ref{eq:Def-E-min-n-plus}) again, this gives the desired result.
\end{proofof}

\section{Proof of Key Step I: Inductive argument for FOEL}
\label{sec:Induct}



In this section we will prove Proposition \ref{prop:inductive1} and
Proposition \ref{prop:inductiveNewLow}.
For Proposition \ref{prop:inductive1} we will actually prove a generalization, where
the generalization is that we do not consider just the Heisenberg model on general graphs with all coupling constants
equal to 1.
But we allow the coupling constants, themselves, to be variable.
%
\begin{definition}
Given a graph $\mathscr{G} = (\mathscr{V},\mathscr{E})$, let us say that $J$
is a ``valid choice of coupling coefficients for $\mathscr{G}$'' if $J$ denotes a function $J : \mathscr{E} \to \R$, such that $J(\{x,y\})$ is nonnegative for each $\{x,y\} \in \mathscr{E}$.
We sometimes say ``$J$ is valid for $\mathscr{G}$.''
\end{definition}
Then we define the Heisenberg Hamiltonian $H(\mathscr{G},J)$ on $\Hil_{\mathscr{V}}$ as
\begin{equation}
\label{eq:HamDefGeneralization}
H(\mathscr{G},J)\, =\, \sum\nolimits_{\{x,y\} \in \mathscr{E}} J(\{x,y\}) h_{xy}\, ,
\end{equation}
where $h_{xy}$ is just as in the original definition (\ref{eq:HamDef}).
We define $\E_n(\mathscr{G},J)$ as
\begin{equation}
\label{eq:EnGJdef}
\E_n(\mathscr{G},J)\, 
=\,
\min\operatorname{spec}\left(H(\mathscr{G},J) \restriction \Hil^{\mathrm{spin}}_{\mathscr{V}}(n)\right)\, ,
\end{equation}
analogously to the definition in (\ref{eq:EDef}).
\begin{definition}
\label{def:FOELJ}
We say that $(\mathscr{G},J)$ satisfies FOEL-$n$ if $\E_n(\mathscr{G},J) \leq \E_{r}(\mathscr{G},J)$ for all $r \geq n$, i.e., for all $r \in \{n,\dots,\lfloor \frac{1}{2}|\mathscr{V}| \rfloor\}$.
\end{definition}

We note that Lemma \ref{lem:preliminary} remains true if we replace $H(\mathscr{G})$ by $H(\mathscr{G},J)$ for a valid choice of couplings.
The analogue of Proposition \ref{prop:LiebMattis} is as follows.
Given $(\mathscr{G},J)$ with $J$ valid for $\mathscr{G}$, define a new graph
$\mathscr{G}_J = (\mathscr{V},\mathscr{E}_J)$ with 
$\mathscr{E}_J = \{\{x,y\} \in \mathscr{E}\, :\, J(\{x,y\})>0\}$.
Then $(\mathscr{G},J)$ satisfies strict FOEL-$0$ if $\mathscr{G}_J$ is connected.
Lieb and Mattis prove all of their results in \cite{LiebMattis} including this one using
general couplings. We merely stated Proposition \ref{prop:LiebMattis} for the special case that 
all couplings are equal to 1
for ease of exposition  in the introduction.
In particular, the analogue of equation (\ref{eq:psdHam}) is true:
if $J$ is a valid choice of coupling coefficients for $\mathscr{G}$, then
\begin{equation}
\label{eq:psdHam1}
H(\mathscr{G},J)\, \geq\, 0\, .
\end{equation}

%

%
The following lemma is a key to the inductive proof of FOEL.
In order to simplify notation, let us begin to use the following natural convention:
\begin{definition}
If $n \in \{1,2,\dots\}$ satisfies $n>\lfloor \frac{1}{2}|\mathscr{V}|\rfloor$, then let us define $\mathfrak{E}_n(\mathscr{G},J)$
to be $+\infty$, interpreting the empty minimum in (\ref{eq:EnGJdef}) as $+\infty$. Similarly, for the case of constant coupling coefficients
$1$, we also define $\mathfrak{E}_n(\mathscr{G})$ to be $+\infty$, in this case that $n$ is larger than the maximum possible value
to get $\Hil^{\mathrm{spin}}_{\mathscr{V}}(n)$ different than $\{0\}$.
\end{definition}
\begin{lemma}
\label{lem:InductiveRefined}
Suppose that $\mathscr{G} = (\mathscr{V},\mathscr{E})$ is a finite graph, and that
$\mathscr{G}' = (\mathscr{V}',\mathscr{E}')$ is a graph such that $\mathscr{V}' = \mathscr{V} \cup \{x'\}$ for a 
single vertex $x'$ not in $\mathscr{V}$, and such that $\mathscr{E} \subseteq \mathscr{E}'$. Let ${J}$ be a valid choice of coupling coefficients for $\mathscr{G}$,
and let ${J}'$ be a valid choice of coupling coefficients for $\mathscr{G}'$.
Suppose, moreover, that
${J}'(\{x,y\}) \geq {J}(\{x,y\})$ for all edges $\{x,y\} \in \mathscr{E}$.
Then $\E_n(\mathscr{G}',J') \geq \min\{\E_n(\mathscr{G},J),\E_{n-1}(\mathscr{G},J)\}$ for each $n\in \{1,\dots,\lfloor \frac{1}{2}|\mathscr{V}'| \rfloor\}$.
\end{lemma}

\begin{proof}
Note that there is a canonical isomorphism $\Hil_{\mathscr{V}'} \cong \Hil_{\mathscr{V}} \otimes \C^2$,
because the last spin site $x'$ also has spin Hilbert space equivalent to $\C^2$.
By addition of angular momenta, using the definition in Lemma \ref{lem:preliminary}, and the formula for addition of angular momenta in $\mathrm{SU}(2)$
\begin{equation}
\label{eq:Addition}
\Hil_{\mathscr{V}'}^{\mathrm{spin}}(n)\, \subseteq \big(\Hil_{\mathscr{V}}^{\mathrm{spin}}(n) \oplus
\Hil_{\mathscr{V}}^{\mathrm{spin}}(n-1)\big) \otimes \C^2\, ,
\end{equation}
for each $n \in \{0,\dots,\lfloor s_{\max}(\mathscr{V}') \rfloor\}$, and where we define
$\Hil_{\mathscr{V}}^{\mathrm{spin}}(-1) = \{0\}$ for consistency, and we define $\Hil_{\mathscr{V}}^{\mathrm{spin}}(n)=\{0\}$ if $n>\lfloor \frac{1}{2}|\mathscr{V}|\rfloor$.
See for example \cite{Edmonds} for a proof of ``addition of angular momenta'' for $\mathrm{SU}(2)$ representations.

Now we may choose a vector $\psi \in \Hil_{\mathscr{V}'}^{\mathrm{spin}}(n)$ with $\|\psi\|=1$ and 
such that $H(\mathscr{G}',J') \psi = E_n(\mathscr{G}',J') \psi$
because $E_n(\mathscr{G}',J')$ is the minimum of the spectrum of $H(\mathscr{G}',J')$ restricted to the subspace 
$\Hil_{\mathscr{V}'}^{\mathrm{spin}}(n)$.
Moreover, by (\ref{eq:Addition}), we may find two orthogonal vectors $\psi_1 \in \Hil_{\mathscr{V}}^{\mathrm{spin}}(n) \otimes \C^2$
and $\psi_2 \in \Hil_{\mathscr{V}}^{\mathrm{spin}}(n-1) \otimes \C^2$ such that $\psi = \psi_1 + \psi_2$.
Furthermore, since $J'(\{x,y\}) \geq J(\{x,y\})$ for all $\{x,y\} \in \mathscr{E}$ (and since $J'(\{x,y\}) \geq 0$ for all $\{x,y\} \in \mathscr{E}' \setminus \mathscr{E}$),
we have, in part using equation (\ref{eq:psdHam1}), 
$$
H(\mathscr{G}',J')\, \geq\, H(\mathscr{G},J) \otimes \mathbbm{1}_{\C^2}\, ,
$$
again, viewing $\Hil_{\mathscr{V}'}$ as equivalent to  $\Hil_{\mathscr{V}} \otimes \C^2$.
Therefore, we have
\begin{equation}
\label{eq:InductiveArgumentStep1}
\E_n(\mathscr{G}',J')\, =\, \langle \psi, H(\mathscr{G}',J') \psi\rangle\,
\geq\,  \langle \psi, \big(H(\mathscr{G},J) \otimes \mathbbm{1}_{\C^2}\big) \psi\rangle\, .
\end{equation}
But both of the subspaces, $\Hil_{\mathscr{V}}^{\mathrm{spin}}(n) \otimes \C^2$ and 
$\Hil_{\mathscr{V}}^{\mathrm{spin}}(n-1) \otimes \C^2$, are invariant subspaces for 
$H(\mathscr{G},J) \otimes \mathbbm{1}_{\C^2}$.
Hence, we may actually write
\begin{equation}
\label{eq:InductiveArgumentStep2}
\langle \psi, \big(H(\mathscr{G},J) \otimes \mathbbm{1}_{\C^2}\big) \psi\rangle\, 
=\, \langle \psi_1, \big(H(\mathscr{G},J) \otimes \mathbbm{1}_{\C^2}\big) \psi_1\rangle
+ \langle \psi_2, \big(H(\mathscr{G},J) \otimes \mathbbm{1}_{\C^2}\big) \psi_2\rangle\, .
\end{equation}
But by the definition of $\E_n(\mathscr{G},J)$ as the minimum of the spectrum of $H(\mathscr{G},J)$ restricted to the subspace 
$\Hil_{\mathscr{V}}^{\mathrm{spin}}(n)$ and similarly for $\E_{n-1}(\mathscr{G},J)$, we have
$$
\langle \psi_1, \big(H(\mathscr{G},J) \otimes \mathbbm{1}_{\C^2}\big) \psi_1\rangle\,
\geq\, \E_n(\mathscr{G},J) \|\psi_1\|^2\, ,\qquad
\langle \psi_2, \big(H(\mathscr{G},J) \otimes \mathbbm{1}_{\C^2}\big) \psi_2\rangle\,
\geq\, \E_{n-1}(\mathscr{G},J) \|\psi_2\|^2\, .
$$
Therefore, combining this with \eqref{eq:InductiveArgumentStep2}, we have
\begin{equation}
\label{eq:InductiveArgumentStep3}
\langle \psi, \big(H(\mathscr{G},J) \otimes \mathbbm{1}_{\C^2}\big) \psi\rangle\, 
\geq\, \min\{\E_r(\mathscr{G},J),\E_{r-1}(\mathscr{G},J)\} (\|\psi_1\|^2 + \|\psi_2\|^2)\, .
\end{equation}
By orthogonality, $\|\psi_1\|^2  + \|\psi_2\|^2 = \|\psi\|^2 = 1$.
So combining  \eqref{eq:InductiveArgumentStep1} with \eqref{eq:InductiveArgumentStep3} gives the result.
\end{proof}

\begin{proposition}
\label{prop:inductive}
Suppose that $\mathscr{G} = (\mathscr{V},\mathscr{E})$ is a finite graph, and that
$\mathscr{G}' = (\mathscr{V}',\mathscr{E}')$ is a graph such that $\mathscr{V}' = \mathscr{V} \cup \{x'\}$ for a 
single vertex $x'$ not in $\mathscr{V}$, and such that $\mathscr{E} \subseteq \mathscr{E}'$. Let ${J}$ be valid for $\mathscr{G}$,
and let ${J}'$ be valid for $\mathscr{G}'$.

Suppose for some $n \in \{0,1,\dots,\lfloor \frac{1}{2}|\mathscr{V}| \rfloor\}$,
that $(\mathscr{G},{J})$ satisfies FOEL-$n$.
Suppose, further that the following two conditions are satisfied:
\begin{itemize}
\item[(1)] ${J}'(\{x,y\}) \geq {J}(\{x,y\})$ for all edges $\{x,y\} \in \mathscr{E}$.
\item[(2)] $\E_{n}(\mathscr{G}',{J}') \leq \E_n(\mathscr{G},{J})$.
\end{itemize}
Then $(\mathscr{G}',{J}')$ satisfies FOEL-$n$.
\end{proposition}
\begin{proof}
Suppose $r\geq n+1$ is fixed.
Applying Lemma \ref{lem:InductiveRefined} (with $n$ in that lemma replaced by $r$),
$$
\E_r(\mathscr{G}',J')\, \geq\, \min\{\E_r(\mathscr{G},J),\E_{r-1}(\mathscr{G},J)\}\, .
$$
But by the assumption that $(\mathscr{G},J)$ satisfies FOEL-$n$, we know that
$$
\min_{r\geq n+1}\min\{\E_r(\mathscr{G},J),\E_{r-1}(\mathscr{G},J)\}\, \geq\, \E_n(\mathscr{G},J)\, .
$$
Therefore, we conclude
$$
\min_{r\geq n+1}\E_r(\mathscr{G}',J')\, \geq\, \min_{r\geq n+1}\min\{\E_r(\mathscr{G},J),\E_{r-1}(\mathscr{G},J)\}\, \geq\,  \E_n(\mathscr{G},J)\,
\geq\, \E_n(\mathscr{G}',J')\, .
$$
But this is a restatement of FOEL-$n$ from Definition \ref{def:FOELJ}.
\end{proof}
Obviously, the proof of Proposition \ref{prop:inductive1} follows.

\begin{proofof}{\bf Proof of Proposition \ref{prop:inductive1}:}
Use Proposition \ref{prop:inductive}, with $J(\{x,y\})=1$ for all $\{x,y\} \in \mathscr{E}$
and $J'(\{x,y\}) = 1$ for all $\{x,y\} \in \mathscr{E}'$.
It is easy to see that the hypotheses of Proposition \ref{prop:inductive} are satisfied for these choices, given
the hypotheses of Proposition \ref{prop:inductive1}.
\end{proofof}
We can also, immediately prove the following corollary
\begin{corollary}
\label{cor:inductive}
For any $n \in \N$ and for any $N \in \{n+1,n+2,\dots\}$, suppose that $\mathscr{G}_{2n},\mathscr{G}_{2n+1},\dots,\mathscr{G}_{N}$ forms a family of graphs, 
and for each $k \in \{2n,\dots,N\}$ that $J_k : \mathscr{E}_k \to \R$ is a valid choice of couplings on $\mathscr{G}_k$,
satisfying the following conditions.
\begin{itemize}
\item[(i)] For each $k \in \{2n,\dots,N\}$, we have that $\mathscr{G}_k = (\mathscr{V}_k,\mathscr{E}_k)$ satisfies $|\mathscr{V}_k|=k$.
\item[(ii)] For each $k \in \{2n,\dots,N-1\}$, we have that 
\begin{itemize}
\item[(a)] $\mathscr{V}_{k} \subset \mathscr{V}_{k+1}$, 
\item[(b)] $\mathscr{E}_k \subset \mathscr{E}_{k+1}$,
\item[(c)] $|\mathscr{V}_{k+1} \setminus \mathscr{V}_k| = 1$, and
\item[(d)] $J_{k+1}(\{x,y\}) \geq J_k(\{x,y\})$ for each $\{x,y\} \in \mathscr{E}_k$.
\end{itemize}
\item[(iii)] For each $k \in \{2n,\dots,N-1\}$, we have that $\E_n(\mathscr{G}_{k+1},J_{k+1}) \leq \E_n(\mathscr{G}_{k},J_k)$. 
\end{itemize}
Then $(\mathscr{G}_k,J_k)$ satisfies FOEL-$n$ according to Definition \ref{def:FOELJ} for each $k \in \{2n,\dots,N\}$.
\end{corollary}
\begin{proof}
The proof is just like the proof of Corollary \ref{cor:inductive1}, except that now we use Proposition \ref{prop:inductive}
everywhere that the old proof used Proposition \ref{prop:inductive1}.
\end{proof}

Next we will prove  Proposition \ref{prop:inductiveNewLow}.
It follows from Corollary \ref{cor:inductive}.

\begin{definition}
\label{def:DilutedSystem}
Suppose that $\mathscr{G}_n,\dots,\mathscr{G}_N$ is an ordered chain of graphs, then we 
say that a valid choice of coupling constants $J_k : \mathscr{E}_k \to [0,\infty)$ for each $k=n,\dots,N$ is a ``diluted system''
if
\begin{itemize}
\item[(i)] ${J}_k \leq \mathbf{1}_{\mathscr{E}_k}$, pointwise, for each $k \in \{n,\dots,N-1\}$, and ${J}_{N} = \mathbf{1}_{\mathscr{E}_N}$;
\item[(ii)] ${J}_{n} \leq \dots\leq {J}_{N}$ in the sense of Proposition \ref{prop:inductive} condition (1).
\end{itemize}
\end{definition}
\begin{lemma}
\label{lem:Zaftig}
If $\mathscr{G}_{2n} ,\dots, \mathscr{G}_N$ is an ordered chain of graphs, and if $N$ is a new low of 
the sequence $(\E_n(\Lambda_k))_{k=2n}^{N}$,
then there is a 
``diluted system'' $J_{2n} \leq \dots \leq J_N$ such that
$(\E_n(\mathscr{G}_k,J_k))_{k=2n}^{N}$ is monotone non-increasing,
and $J_N = \mathbf{1}_{\mathscr{E}_N}$. In other words, $H(\mathscr{G}_N,J_N) = H(\mathscr{G}_N)$.
\end{lemma}
\begin{proofof}{\bf Proof of Lemma \ref{lem:Zaftig}:}
We will actually construct a diluted system $J_{2n},\dots,J_N$ satisfying
\begin{equation}
\label{eq:MinHull}
\E_n(\mathscr{G}_k,J_k) = \min\{\E_n(\mathscr{G}_i)\, :\, i =2n,\dots,k\}\, ,
\end{equation}
for each $k = 2n,\dots,N$.

The proof is by induction. 
For the initial step, we let $J_{2n} = \mathbf{1}_{\mathscr{E}_{2n}}$.
Then \eqref{eq:MinHull} follows, trivially.

Suppose that some diluted system $J_{2n},\dots,J_{k}$ has been constructed up to $k$, for some $k \in \{2n,\dots,N-1\}$,
satisfying (\ref{eq:MinHull}).
For each $t \in [0,1]$, we define a choice of valid coupling constants $J_{k+1}^{(t)} : \mathscr{E}_{k+1} \to [0,\infty)$
as follows:
$$
J^{(t)}_{k+1}(\{x,y\})\,
=\, \begin{cases} (1-t) J_k(\{x,y\}) + t & \text{ for $\{x,y\} \in \mathscr{E}_k$,}\\
t & \text{ for $\{x,y\} \in \mathscr{E}_{k+1} \setminus \mathscr{E}_k$,}\\
0 & \text{ for $\{x,y\} \not\in \mathscr{E}_{k+1}$.}
\end{cases}
$$
Note that $J^{(0)}_{k+1}$ is $J_k$ viewed as a function on $\mathscr{E}_{k+1}$, as one can see from the formula.
The couplings are pointwise non-decreasing in $t$.
Also, $J^{(1)}_{k+1} = \mathbf{1}_{\mathscr{E}_{k+1}}$.

Because of the last fact, viewing $\Hil_{\mathscr{V}_{k+1}}$ as being canonically isomorphic to $\Hil_{\mathscr{V}_k} \otimes \C^2$,
we have
$$
H(\mathscr{G}_{k+1},J^{(0)}_{k+1})\, =\, H(\mathscr{G}_k,J_k) \otimes \mathbbm{1}_{\C^2}\, .
$$
We may choose a vector $\phi \in \Hil_{\mathscr{V}_k}(n,n)$ such that $\|\phi\|=1$ and $H(\mathscr{G}_k,J_k) \phi = \E_n(\mathscr{G}_k,J_k) \phi$.
Then we may take $\psi = \phi \otimes \ket{\uparrow}$ in  $\Hil_{\mathscr{V}_{k+1}}$.
This will be in $\Hil_{\mathscr{V}_{k+1}}(n,n)$, and $H(\mathscr{G}_{k+1},J^{(0)}_{k+1}) \psi = \E_n(\mathscr{G}_k,J_k) \phi$.
These facts are non-trivial. But we leave them to the reader as an exercise in the definitions of magnetization and spin, and an application
of ``addition of angular momentum.''
This proves that 
\begin{equation}
\label{ineq:UpperZeroPoint}
\E_n(\mathscr{G}_{k+1},J^{(0)}_{k+1})\, \leq\, \E_n(\mathscr{G}_k,J_k)\, .
\end{equation}
There are two cases to consider.

In Case 1, we have the condition $\E_n(\mathscr{G}_{k+1}) \leq \E_n(\mathscr{G}_k,J_k)$. Then we take $J_{k+1} = J^{(1)}_{k+1}$. In other words,
in that case, we choose $t=1$. This means $J_{k+1} = \mathbf{1}_{\mathscr{E}_{k+1}}$ in this case.
That means that $\E_n(\mathscr{G}_{k+1},J_{k+1}) = \E_n(\mathscr{G}_{k+1})$.

In Case 2 we have $\E_n(\mathscr{G}_{k+1}) > \E_n(\mathscr{G}_k,J_k)$.
Then we know the following.
The mapping $t \mapsto \E_n(\mathscr{G}_{k+1},J_{k+1}^{(t)})$ is non-decreasing and continuous in $t$.
At $t=0$ we get a number at most equal to $\E_n(\mathscr{G}_k,J_k)$ according to \eqref{ineq:UpperZeroPoint}.
At $t=1$ we get $\E_n(\mathscr{G}_{k+1})$, a number strictly greater than $\E_n(\mathscr{G}_k,J_k)$.
Therefore, there is at least one intermediate value $t \in [0,1)$ with $\E_n(\mathscr{G}_{k+1},J_{k+1}^{(t)}) =\E_n(\mathscr{G}_k,J_k)$.
We let $\mathscr{T}$ be the set of all such $t$'s, and we let $t_* = \sup(\mathscr{T})$. By continuity of the mapping, $t_* \in \mathscr{T}$.
We let $J_{k+1} = J_{k+1}^{(t^*)}$.
In this case, we have $\E_n(\mathscr{G}_{k+1},J_{k+1}) =\E_n(\mathscr{G}_k,J_k)$.

In Case 1, we assumed $\E_n(\mathscr{G}_{k+1}) \leq \E_n(\mathscr{G}_k,J_k)$ and obtained $J_{k+1}$ such that $\E_n(\mathscr{G}_{k+1},J_{k+1}) = \E_n(\mathscr{G}_{k+1})$.
In Case 2, we assumed $\E_n(\mathscr{G}_{k+1}) > \E_n(\mathscr{G}_k,J_k)$ and obtained $J_{k+1}$ such that 
$$
\E_n(\mathscr{G}_{k+1},J_{k+1})\, =\, \E_n(\mathscr{G}_k,J_k)\, .
$$
Therefore, in either case, we have
$$
\E_n(\mathscr{G}_{k+1},J_{k+1})\, =\, \min\{\E_n(\mathscr{G}_k,J_k),\E_n(\mathscr{G}_{k+1})\}\, .
$$
By the induction hypothesis and (\ref{eq:MinHull}) we see that in fact
$$
\E_n(\mathscr{G}_{k+1},J_{k+1})\, =\, \min\{\E_n(\mathscr{G}_{2n}),\dots,\E_n(\mathscr{G}_{k+1})\}\, .
$$
But this is precisely \eqref{eq:MinHull} with $k$ replaced by $k+1$. So the induction step is proved.
\end{proofof}

\begin{proofof}{\bf Proof of Proposition \ref{prop:inductiveNewLow}:}
By Lemma \ref{lem:Zaftig} we know that there exists a diluted system $J_{2n},\dots,J_{N}$ such that $H(\mathscr{G}_N,J_N) = H(\mathscr{G}_N)$.
Then, by Corollary \ref{cor:inductive}, we see that $(\mathscr{G}_k,J_k)$ satisfies FOEL-$n$, for each $k=2n,\dots,N$.
Therefore, we have, for any $k \in \{2n,\dots,N\}$, 
$$
\min_{r\geq n}\E_r(\mathscr{G}_k,J_k)\, \geq\, \E_n(\mathscr{G}_k,J_k)\, .
$$
But since $(\E_n(\mathscr{G}_k,J_k))_{k=2n}^{N}$ is non-increasing, this means
\begin{equation}
\label{ineq:slashE}
\min_{r\geq n}\E_r(\mathscr{G}_k,J_k)\, \geq\, \E_n(\mathscr{G}_k,J_k)\, \geq\, \E_n(\mathscr{G}_N,J_N)\, =\, \E_n(\mathscr{G}_N)\, ,
\end{equation}
where the last equality holds because $H(\mathscr{G}_N,J_N) = H(\mathscr{G}_N)$.
Moreover, since $J_k \leq \mathbf{1}_{\mathscr{E}_k}$, we know that as operators $H(\mathscr{G}_k,J_k) \leq H(\mathscr{G}_k)$,
meaning the difference is positive semi-definite.
This is the type of ordering which leads to inequalities for the Rayleigh quotients. 
Therefore, we also have $\E_r(\mathscr{G}_k,J_k) \leq \E_r(\mathscr{G}_k)$ for all $r$.
Putting this together with \eqref{ineq:slashE}, we obtain
$$
\E_n(\mathscr{G}_N)\, \leq\, \E_r(\mathscr{G}_k,J_k)\, \leq\, \E_r(\mathscr{G}_k)\, ,
$$
for each $k \in \{2n,\dots,N\}$ and each $r \in \{n,\dots,\lfloor k/2 \rfloor\}$.

That is what we wanted to prove.
\end{proofof}

\section{Proof of Key Step  II -- Part A: Linear Spin Wave Approximation}
\label{sec:LSW}

We now wish to prove Key Step II, meaning
Theorem \ref{thm:SpecGather}.
We also needed Proposition \ref{prop:Variational1} in the conditional proof of the main theorem.
But we will prove Proposition \ref{prop:Variational1} in Appendix \ref{app:Variational1}.
It will be easy, but it will follow a direct calculation using the explicit formula for the approximate eigenvectors.

Theorem \ref{thm:SpecGather} and Proposition \ref{prop:Variational1}
 give an approximation for the low-energy spectrum of the
ferromagnetic Heisenberg Hamiltonian $H_{\Lambda(d,N)}$ in the limit $N\to\infty$.
This approximation is part of the linear spin wave approximation:
it is the specialization of the linear spin wave approximation at lowest energies,
which are energies bounded by a finite multiple of the spectral gap.
We will complete this part in several steps.
\begin{itemize}
\item
In this section, Part A, we introduce the transformation needed to compare to an ideal Bose gas.
We also state the main results about this transformation and an approximate inverse, showing that these transformations do not distort low-energy
trial wave functions, too much.
This relies on two important types of bounds.
\item
In Part B, we give variational arguments, combined with Chebyshev's inequality, to apply this comparison to prove 
Theorem \ref{thm:SpecGather}.
To apply Chebyshev's inequality, we will also need  Proposition \ref{prop:Variational1}. But this proposition will be proved in Appendix \ref{app:Variational1}
\item
In Part C, we prove one of the two important types of bounds, which is a trace-theorem type bound.
It shows that, relative to the ideal Bose gas, a variational trial wave vector with low energy does not lose much norm
in the contraction transformation to the quantum spin system.
\item
In Part D, which is an extension theorem, we prove a complementary result: given a low energy
trial wave vector for the quantum spin system, we may extend it to a wave vector for the ideal
Bose gas such that we raise the energy only by a bounded factor.
\end{itemize}
The linear spin wave approximation goes back to Bloch \cite{Bloch},
who used the unproven hypothesis to derive a non-rigorous formula for the low temperature pressure.
There were further developments by Dyson in \cite{Dyson1,Dyson2} who showed there must be corrections to the naive approach.
The first rigorous analysis was initiated by Conlon and Solovej 
who obtained one-sided bounds for the pressure  \cite{ConlonSolovej1,ConlonSolovej2}.
The bounds were improved by Balint Toth \cite{Toth}. 
Improtantly, Toth also introduced a new representation
for the quantum Heisenberg ferromagnet in terms of the interchange process\footnote{
The interchange process is also known as the stirring process in \cite{Liggett},
where it is discussed especially as the graphical representation for the symmetric exclusion process.
Also, coincidentally, Larry Thomas had noted a relation between the Heisenberg model and the symmetric
exclusion process in \cite{Thomas}, although he did not capitalize on this observation.}.

More recently, the linear spin wave approximation was reconsidered by Correggi, Giuliani and Seiringer
who improved Conlon and Solovej's and Toth's bounds, to verify Bloch's original ansatz for the asymptotic formula for the free energy
at low temperatures \cite{CorreggiGiulianiSeiringer}.
Their article is a good reference for the linear spin wave approximation.

For $n \in \{0,1,\dots\}$, 
let $\ell^2(\mathscr{V}^n)$ denote the vector space (of dimension $|\mathscr{V}|^n$)
of functions $F : \mathscr{V}^n \to \C$, with the usual $\ell^2$ norm
$
\|F\|^2 = \sum_{(x_1,\dots,x_n) \in \mathscr{V}^n} |F(x_1,\dots,x_n)|^2
$.
Then, 
for each $k \in \{1,\dots,n\}$, we denote
the $k$th particle graph Laplacian as
$\mathfrak{H}^{(n)}_{\mathscr{G},k}:\ell^2(\mathscr{V}^n)\to \ell^2(\mathscr{V}^n)$, where, for each $F \in \ell^2(\mathscr{V}^n)$
and each $(x_1,\dots,x_n) \in \mathscr{V}^n$,
\begin{equation}
\label{eq:Hfree}
\mathfrak{H}^{(n)}_{\mathscr{G},k} F(x_1,\dots,x_n)\,
=\, \frac{1}{2}\,
\sum_{y\in\mathcal{N}(\mathscr{E},x_k)} (F(x_1,\dots,x_n)
- F(x_1,\dots,x_{k-1},y,x_{k+1},\dots,x_n))\, .
\end{equation}
Given any $\pi \in S_n$, we may define the operator $U_{\mathscr{V}}(\pi) : \ell^2(\mathscr{V}^n) \to \ell^2(\mathscr{V}^n)$ such that
$$
U_{\mathscr{V}}(\pi) F(x_1,\dots,x_n)\, =\, F(x_{\pi^{-1}(1)},\dots,x_{\pi^{-1}(n)})\, ,
$$
for each $F \in \ell^2(\mathscr{V}^n)$
and each $(x_1,\dots,x_n) \in \mathscr{V}^n$.
\begin{lemma}
\label{lem:free}
(a) For each $\pi \in S_n$, the operator $U_{\mathscr{V}}(\pi)$ is unitary.\\
(b) For each $k \in \{1,\dots,n\}$ and each $\pi \in S_n$,
$$
U_{\mathscr{V}}(\pi)^* \mathfrak{H}^{(n)}_{\mathscr{G},k} U_{\mathscr{V}}(\pi)\, 
=\, \mathfrak{H}^{(n)}_{\mathscr{G},\pi^{-1}(k)}\, .
$$
\end{lemma}
We will prove this easy lemma in Appendix \ref{app:Inter}.

The total $n$-particle Hamiltonian for $n$ free particles on $\mathscr{G}$ is $\mathfrak{H}^{(n)}_{\mathscr{G}} : \ell^2(\mathscr{V}^n) \to \ell^2(\mathscr{V}^n)$,
given by
$$
\mathfrak{H}^{(n)}_{\mathscr{G}}\, =\, \sum_{k=1}^{n} \mathfrak{H}^{(n)}_{\mathscr{G},k}\, .
$$
Since this is the uniform sum, 
$\mathfrak{H}^{(n)}_{\mathscr{G}}$ commutes with $U_{\mathscr{V}}(\pi)$
for each $\pi \in S_n$.
We denote the symmetrization projection as $\mathfrak{S}_{\mathscr{V}}^{(n)} : \ell^2(\mathscr{V}^n) \to \ell^2(\mathscr{V}^n)$,
$$
\mathfrak{S}_{\mathscr{V}}^{(n)}\, =\, \frac{1}{n!}\, \sum_{\pi \in S_n} U_{\mathscr{V},\pi}\, .
$$
Hence, $\mathfrak{H}^{(n)}_{\mathscr{G}}$ commutes with $\mathfrak{S}_{\mathscr{V}}^{(n)}$.
\begin{remark}
\label{rem:FirstGraph}
Given the graph $\mathscr{G} = (\mathscr{V},\mathscr{E})$, we may define two different graph structures on $\mathscr{V}^n$.
Let $\Phi_n(\mathscr{E})$ denote the set of all two-element sets $\{
(x_1,\dots,x_n),(x_1,\dots,x_{k-1},y,x_{k+1},\dots,x_n)\}$ for all possible choices of
$(x_1,\dots,x_n) \in \mathscr{V}^n$, $k \in \{1,\dots,n\}$
and $y \in \mathscr{N}(\mathscr{E},x_k)$.
Then the graph Laplacian for $(\mathscr{V}^n,\Phi_n(\mathscr{E}))$ is $\mathfrak{H}^{(n)}_{\mathscr{G}}$,
using the definition as in (\ref{eq:GraphLap}).
\end{remark}
Considering the remark, there is another important edge set.
Let
\begin{equation}
\label{eq:InDefinition}
\mathscr{I}_n(\mathscr{V})\,
=\, \{(x_1,\dots,x_n) \in \mathscr{V}^n\, :\,
\exists i,j \in \{1,\dots,n\}\ \text{ s.t. }\ i\neq j\ \text{ and }\ x_i=x_j\}\, .
\end{equation}
Then we may define $\Theta_n(\mathscr{E})$ to denote the set of all  edges in $\Phi_n(\mathscr{E})$ satisfying the 
additional condition that
neither endpoint is in $\mathscr{I}_n(\mathscr{V})$.
Then the graph $(\mathscr{V}^n\setminus \mathscr{I}_n(\mathscr{V}),\Theta_n(\mathscr{V}))$ has a Laplacian
closely connected $H_{\mathscr{G}} \restriction \Hil^{\mathrm{mag}}_{\mathscr{V}}(n)$.

If we restrict $\mathfrak{H}^{(n)}_{\mathscr{G}}$ to $\mathfrak{S}_{\mathscr{V}}^{(n)}(\ell^2(\mathscr{V}^n))$,
then this is the Hamiltonian for $n$ particles in an ideal Bose gas on $\mathscr{G}$.
Part of the linear spin wave approximation  states that this is a ``good'' approximation for the ferromagnetic Heisenberg
Hamiltonian $H_{\mathscr{G}}$ restricted to $\Hil^{\mathrm{mag}}_{\mathscr{V}}(n)$,
in a certain sense.

Recall the mapping $T_{\mathscr{V}}^{(n)} : \ell^2(\mathscr{V}^n) \to \Hil_{\mathscr{V}}$, from Definition \ref{eq:defTn0}.
The following is easy to prove.
\begin{lemma}
\label{lem:monoton}
(a)
For each $n \in \{1,\dots,|\mathscr{V}|\}$,
and each $F \in \ell^2(\mathscr{V}^n)$,
$$
\|T_{\mathscr{V}}^{(n)} F\|\, \leq\, \|F\|\, .
$$
(b) For each $n \in \{1,\dots,|\mathscr{V}|\}$, the operator $\mathfrak{H}^{(n)}_{\mathscr{G}}$ on $\ell^2(\mathscr{V}^n)$ is positive semi-definite, and
$$
T_{\mathscr{V}}^{(n)}\mathfrak{H}^{(n)}_{\mathscr{G}}\big(T_{\mathscr{V}}^{(n)}\big)^*\,
\geq\, \big(H_{\mathscr{G}} \restriction \Hil^{\mathrm{mag}}_{\mathscr{V}}(n)\big)\, .
$$
\end{lemma}
We will also prove this easy result in  Appendix \ref{app:Inter},
since it follows a similar type of argument as the proof of Lemma \ref{lem:free}.
Part (b) will follow from the fact already expressed in Remark \ref{rem:FirstGraph}, as well as the fact, which will be used
several more times, that the mapping from $\mathscr{G}$ to $-\Delta_{\mathscr{G}}$, as defined in equation (\ref{eq:GraphLap}),
is increasing,  relative to the cone of positive semi-definite operators, as we increase edges.
(This was also a key to the arguments in Section \ref{sec:Induct}.)
When we speak about this fact, now, we are usually thinking of the edge set as being $\Phi_n(\mathscr{E})$
or $\Theta_n(\mathscr{E})$, as opposed to the edge set of the original graph.

To prove a part of the linear spin wave approximation, we need some complementary
results, as well.
For these inequalities, we revert to considering the special family of graphs from before.
\begin{definition}
Fix, $d,n \in \{1,2,\dots\}$.
For each $N \in \{1,2,\dots\}$, define $\widetilde{T}^{(n)}_{d,N} : \ell^2(\B^d(L^+(d,N))) \to \Hil^{\mathrm{mag}}_{\Lambda(d,N)}(n)$ such that for any 
$F \in \ell^2(\B^d(L^+(d,N)))$, we first take the restriction $F \restriction (\Lambda(d,N))^n$, restricting to the domain of $T^{(n)}_{\Lambda(d,N)}$,
and then we take $\widetilde{T}^{(n)}_{d,N} F = T^{(n)}_{\Lambda(d,N)} \big(F \restriction (\Lambda(d,N))^n\big)$.
\end{definition}
A trivial corollary of Lemma \ref{lem:monoton} is this
\begin{lemma} 
\label{lem:contractionTildeT}
For any $d,L \in \{1,2,\dots\}$ and $n \in \{0,1,\dots\}$, we have that $\|\widetilde{T}_{d,N}^{(n)}\|\leq 1$ and 
\begin{equation}
\label{eq:contractionTildeT}
\widetilde{T}_{d,N}^{(n)}\mathfrak{H}^{(n)}_{\B^d(L^+(d,N))}\big(\widetilde{T}_{d,N}^{(n)}\big)^*\,
\geq\, \big(H_{\Lambda(d,N)} \restriction \Hil^{\mathrm{mag}}_{\Lambda(d,N)}(n)\big)\, .
\end{equation}
\end{lemma}
Once again, we will prove this simple lemma in Appendix \ref{app:Inter}.
But the idea is just to use the fact, again, that increasing edges of a graph $\mathscr{G}$ increases the graph Laplacian $-\Delta_{\mathscr{G}}$,
relative to the psd cone.
There are clearly at least as many edges in $\B^d(L^+(d,N))$ as in $\Lambda(d,N)$ since $\Lambda(d,N) \subseteq \B^d(L^+(d,N))$.

\begin{proposition}
\label{prop:trace0}
Fix, $d,n \in \{1,2,\dots\}$.
Then there is an $N_1(d,n) \in \{1,2,\dots\}$ and two constants $C_1(d,n)$, $C_2(d,n)$ such that,
for any $N \geq N_1(d,n)$ and for
any $F \in \mathfrak{S}_{\B^d(L^+(d,N))}^{(n)}(\ell^2(\B^d(L^+(d,N))^n))$,
we have
\begin{equation}
\label{eq:deformTilldeT}
\|\widetilde{T}^{(n)}_{d,N} F\|^2\, \geq\, \left(1 - \frac{C_1(d,n)}{L^+(d,N)}\right)\|F\|^2
- C_2(d,N) L^+(d,N) \langle F, \mathfrak{H}^{(n)}_{\Lambda(d,N)} F\rangle\, .
\end{equation}
\end{proposition}
We will prove Proposition \ref{prop:trace0} in Section \ref{sec:trace0}.
\begin{proposition}
\label{prop:extension0}
Fix $d,n\in \{1,2,\dots\}$.
There exists an $N_2(d,n) \in \{1,2,\dots\}$ and a constant $C_3(d,n)$ such that,
for each $N \in \{1,2,\dots\}$ satisfying $N\geq N_3(d,n)$, there is a linear mapping 
$$
\Xi^{(n)}_{d,N} : \Hil^{\mathrm{mag}}_{\Lambda(d,N)}(n) \to \mathfrak{S}^{(n)}_{\B^d(L^+(d,N))}\Big(\ell^2\Big(\big(\B^d(L^+(d,N))\big)^n\Big)\Big)\, ,
$$
such that $\widetilde{T}^{(n)}_{d,N} \cdot \Xi^{(n)}_{d,N}$ is the identity on $\Hil^{\mathrm{mag}}_{\Lambda(d,N)}(n)$
and such that
\begin{equation}
\label{eq:extBd}
\big\langle \Xi^{(n)}_{d,N} \Psi, \mathfrak{H}_{\B^d(L^+(d,N))}^{(n)} \Xi^{(n)}_{d,N} \Psi \big\rangle\,
\leq\, C_3(d,n) \langle \Psi, H_{\Lambda(d,N)} \Psi\rangle\, ,
\end{equation}
for each $\Psi \in \Hil^{\mathrm{mag}}_{\Lambda(d,N)}(n)$.
\end{proposition}
We will prove this proposition in Section \ref{sec:extension0}.

In the next section, we will explain how to 
prove Theorem \ref{thm:SpecGather}.
We do thi by combining Proposition \ref{prop:Variational1}, Proposition \ref{prop:trace0}, Proposition {prop:extension0}, 
the Rayleigh-Ritz variational approach and Chebyshev's inequality.

\section{Proof of Key Step II -- Part B:  Variational Argument}

Let us begin by noting the formula for the spectrum of the non-interacting Bose gas.
\begin{definition}
For $d,L \in \{1,2,\dots\}$, let
$\mathscr{O}(d,L)$ denote the set of all ``occupation functions,''
defined to be functions $\nu : \{0,\dots,L-1\}^d \to \{0,1,\dots\}$.
For $n \in \{0,1,\dots\}$ let $\mathscr{O}_n(d,L)$ be the set of $\nu \in \mathscr{O}(d,L)$ such that
$$
\sum_{\boldsymbol{\kappa} \in \{0,\dots,L-1\}^d} \nu(\boldsymbol{\kappa})\, =\, n\, .
$$
\end{definition}
For each choice of $\nu \in \mathscr{O}_n(d,L)$, let us define
$\mathfrak{K}_{d,L}(\nu)$ to be a subset of $(\{0,\dots,L-1\}^d)^n$, defined as the set of all 
$(\boldsymbol{\kappa}_1,\dots,\boldsymbol{\kappa}_n)$
such that
$$
\forall \boldsymbol{\kappa}
\in \{0,\dots,L-1\}^d\, :\
|\{k \in \{1,\dots,n\}\, :\, \boldsymbol{\kappa}_k=\boldsymbol{\kappa}\}| = \nu(\boldsymbol{\kappa})\, .
$$
\begin{definition}
\label{def:WidetildeF}
For each $\nu \in \mathscr{O}_n(d,L)$, define a function $\widetilde{F}^{(n)}_{d,L}(\nu;\cdot) \in \ell^2((\B^d(L))^n)$
by the formula
\begin{equation}
\widetilde{F}^{(n)}_{d,L}(\nu;\r_1,\dots,\r_n)\,
=\, L^{-nd/2} |\mathfrak{K}_{d,L}(\nu)|^{-1/2} 
\sum_{(\boldsymbol{\kappa}_1,\dots,\boldsymbol{\kappa}_n) \in \mathfrak{K}_{d,L}(\nu)}
\prod_{k=1}^{n} \prod_{j=1}^{d} f(L^{-1} \kappa_{k,j}, r_{k,j})\, ,
\end{equation}
for each $(\r_1,\dots,\r_n) \in (\B^d(L))^n$, where $f(\xi,r)$ for $\xi \in \R$ and $r \in \Z$ is given by Definition \ref{def:littleF}.
\end{definition}
The spectrum of $\mathfrak{H}^{(n)}_{\B^d(L)}$ restricted to the invariant subspace $\mathfrak{S}^{(n)}_{\B^d(L)}(\ell^2((\B^d(L))^n))$
may then be summarized, as follows.
\begin{lemma}
\label{lem:specSummary}
For $d,L \in \{1,2,\dots\}$ and $n \in \{0,1,\dots\}$, an orthonormal basis of $\mathfrak{S}^{(n)}_{\B^d(L)}(\ell^2((\B^d(L))^n))$
is $\big(\widetilde{F}^{(n)}_{d,L}(\nu;\cdot)\, :\, \nu \in \mathscr{O}_n(d,L)\big)$, and 
\begin{equation}
\mathfrak{H}^{(n)}_{\B^d(L)} \widetilde{F}^{(n)}_{d,L}(\nu;\cdot)\,
=\, \sum_{\k \in \{0,\dots,L-1\}^d} \nu(\k) \sum_{j=1}^{d} 2 \sin^2\left(\frac{\pi \kappa_{j}}{2L}\right)  \widetilde{F}^{(n)}_{d,L}(\nu;\cdot)\, ,
\end{equation}
for each $ \nu \in \mathscr{O}_n(d,L)$.
\end{lemma}
This easy lemma will also be proved in Appendix \ref{app:specSummary}.  This calculation is well-known, since it is the spectrum
of a non-interacting ideal Bose gas on a finite box in a lattice. But we will include it for completeness.

We may now prove part of Proposition \ref{prop:Variational1}.
First, given $(\k_1,\dots,\k_n) \in \mathcal{K}(d,n)$, if $L$ is sufficiently large that $\k_k \in \{0,\dots,L-1\}^d$ for each $k \in \{1,\dots,n\}$, then let us define
$\nu_{(\k_1,\dots,\k_n)}$ to be the associated point of $\mathscr{O}_n(d,L)$:
\begin{equation}
\label{eq:nukappa}
\nu_{(\k_1,\dots,\k_n)}(\k)\, =\, \sum_{k=1}^{n} \mathbf{1}_{\{\k_k=\k\}}\, ,
\end{equation}
for each $\k \in \{0,\dots,L-1\}^d$.
A calculation then shows 
\begin{equation}
\label{eq:IdentityTildePsiTildeF}
\widetilde{\Psi}_{d,N}^{(n)}(\k_1,\dots,\k_n)\, =\, n! |\mathfrak{K}_{d,L^+(d,N)}(\nu_{(\k_1,\dots,\k_n)})|^{-1/2}
\widetilde{T}^{(n)}_{d,N} \widetilde{F}^{(n)}_{d,L^+(d,N)}(\nu_{(\k_1,\dots,\k_n)};\cdot)\, .
\end{equation}

Next we want to establish variational upper bounds, and then lower bounds using Chebyshev's inequality.
For the variational upper bounds we will use the following:
For a self-adjoint operator $A$ on a $d$-dimensional Hilbert space $\Hil$, we have that
\begin{equation}
\label{eq:RR}
\lambda_k\, =\, \min\left\{\max_{v \in \operatorname{span}(\{v_1,\dots,v_k\})\setminus \{0\}} \frac{\langle v, A v\rangle}{\|v\|^2}\, :\, v_1,\dots,v_k \in \Hil\text{ are linearly independent}\right\}\, ,
\end{equation}
for each $k \in \{1,\dots,d\}$, where the eigenvalues of $A$ are enumerated in non-decreasing order $\lambda_1 \leq \dots\leq \lambda_d$, repeating according to multiplicity.
With this, we may prove the first part of the arguments.
\begin{lemma}
\label{lem:UpperBdVar}
Let $\mathcal{S}$ be a finite subset of $\widetilde{\mathcal{K}}(d,n)$. Let $|\mathcal{S}|$  denote the cardinality of $\mathcal{S}$, as usual.
But also let us introduce a new notation
$$
\|\mathcal{S}\|^2\, =\, \max\Big\{ \sum_{k=1}^{n} \sum_{j=1}^{d} \kappa_{k,j}^2\, :\,
[(\k_1,\dots,\k_n)]_{S_n} \in \mathcal{S}\Big\}\, .
$$
(The sum is permutation independent, so the definition does not depend on which representatives of each equivalence class are chosen.)
Recall the definition of the spectral subspaces from (\ref{eq:espace}).
For each $\epsilon>0$, there exists an $N_0(\mathcal{S},\epsilon)$ such that 
$$
\dim\left(\mathfrak{L}_{\Lambda(d,N)}\left(0,\gamma \cdot [L^+(d,N)]^{-2}(\|\mathcal{S}\|^2+\epsilon)\right)\right)\, \geq\, |\mathcal{S}|\, .
$$
\end{lemma}
\begin{proof}
We consider the associated vectors $\widetilde{\Psi}_{d,N}^{(n)}(\k_1,\dots,\k_n)$,
one chosen for each $[(\k_1,\dots,\k_n)]_{S_n} \in \mathcal{S}$ 
(and noting that the definition is permutation independent, so that the choice of which
representative we choose is not important).
By equation (\ref{eq:AsymptoticIP}), we know
they are linearly independent, for sufficiently large $N$.
For any element of the span, we have some $c : \mathcal{S} \to \C$ such that the vector is 
$$
\Phi\, =\,
\sum_{[(\k_1,\dots,\k_n)]_{S_n} \in \mathcal{S}} c([(\k_1,\dots,\k_n)]_{S_n}) 
\widetilde{\Psi}_{d,N}^{(n)}(\k_1,\dots,\k_n)
$$
Let us also define
$$
G\, =\, 
\sum_{[(\k_1,\dots,\k_n)]_{S_n} \in \mathcal{S}} c([(\k_1,\dots,\k_n)]_{S_n}) 
\widetilde{F}_{d,L^+(d,N)}^{(n)}(\nu_{(\k_1,\dots,\k_n)};\cdot)\, .
$$
Then, using equation (\ref{eq:AsymptoticIP}) from Proposition \ref{prop:Variational1}, that for any $\epsilon>0$ and sufficiently large $N$ we have
\begin{equation}
\label{eq:asymptoticExact}
1\, \geq\, \frac{\|\Psi\|^2}{\sum_{[(\k_1,\dots,\k_n)]_{S_n} \in \mathcal{S}} |c([(\k_1,\dots,\k_n)]_{S_n})|^2
(n!)^2 |\mathfrak{K}_{d,L^+(d,N)}(\nu)|^{-1} }\, \geq\, 1 - \epsilon\, ,
\end{equation}
The denominator in the middle expression is equal to $\|G\|^2$. 
The bounds in (\ref{eq:asymptoticExact}) are both upper and lower bounds for the denominator term in (\ref{eq:RR}).

But for the numerator, we have just upper bounds, so far.
Using (\ref{eq:contractionTildeT}) in Lemma \ref{lem:contractionTildeT}, and using (\ref{eq:IdentityTildePsiTildeF}), we have
\begin{align*}
\langle \Psi, H_{\Lambda(d,N)} \Psi\rangle\, 
&\leq\, \langle G, \mathfrak{H}^{(n)}_{\B^d(L^+(d,N))} G\rangle\\
&=\,
\sum_{[(\k_1,\dots,\k_n)]_{S_n} \in \mathcal{S}} |c([(\k_1,\dots,\k_n)]_{S_n})|^2
(n!)^2 |\mathfrak{K}_{d,L^+(d,N)}(\nu)|^{-1} \sum_{k=1}^{n} \sum_{j=1}^{d} \kappa_{k,j}^2\, .
\end{align*}
Then the result follows by a max-norm bound of the final sum on the right-hand-side using $\|\mathcal{S}\|^2$ to replace
$\sum_{k=1}^{n} \sum_{j=1}^{d} \kappa_{k,j}^2$ (and (\ref{eq:asymptoticExact}) for the denominator).
\end{proof}

Now we may complete the proof of 
Theorem \ref{thm:SpecGather}.

\begin{proofof}{\bf Proof of Theorem \ref{thm:SpecGather}:}
We now use Proposition \ref{prop:extension0}.
Suppose that $m$ is fixed, as well as a small parameter $\epsilon \in (0,1/2)$,
and that we have $R$ orthonormal eigenvectors of $H_{\Lambda(d,N)}$, called 
$\Phi_1,\dots,\Phi_R \in \mathfrak{S}^{(n)}_{\B^d(L^+(d,N))} (\ell^2((\B^d(L^+(d,N)))^n))$, with energy eigenvalues
equal to $\gamma \cdot [L^+(d,N)]^{-2} \lambda_r$ for each $r \in \{1,\dots,R\}$.
And suppose that the rescaled eigenvalues satisfy $\lambda_1,\dots,\lambda_R \in [0,m+\epsilon]$.
Let $G_r = \Xi^{(n)}_{d,N} \Phi_r$ for each $r \in \{1,\dots,R\}$.
Then by Proposition \ref{prop:extension0}, we know that
\begin{equation}
\label{eq:GrEnBd}
\big\langle G_r\, ,\ \mathfrak{H}^{(n)}_{\B^d(L^+(d,N))} G_r\big\rangle\, \leq\, \gamma \cdot [L^+(d,N)]^{-2} (m+\epsilon) C_3(d,n)\, ,
\end{equation}
where $C_3(d,n)$ is a large but fixed constant, independent of $N$.
We also know that $\widetilde{T}^{(n)}_{d,N} G_r = \Phi_r$ for each $r$. In particular, this means that 
$$
\|G_r\|^2\, \geq\, \|\widetilde{T}^{(n)}_{d,N} G_r \|^2\, =\, \|\Phi_r\|^2\, =\, 1\, ,
$$
by equation (\ref{eq:contractionTildeT}).
Then, using Lemma \ref{lem:specSummary} and Chebyshev's inequality, we may write
$$
G_r\, =\, \sum_{\nu \in \mathscr{O}_n(d,L^+(d,N))} c_r(\nu) \widetilde{F}^{(n)}_{d,L^+(d,N)}(\nu;\cdot)\, ,
$$
and we may bound the coefficients of $\nu$'s for which $\widetilde{F}^{(n)}_{d,L^+(d,N)}(\nu;\cdot)$ has a large eigenvalue.
More precisely, let us use the bound $\sin(\theta) \geq 2\theta/\pi$ for $\theta \in [0,\pi/2]$, and let us define, for a given $M$,
$$
\mathcal{S}_M\, =\, \Big\{\nu \in \mathscr{O}_n(d,L^+(d,N))\, :\, \sum_{\k \in \{0,\dots,L-1\}^d} \nu(\k) \sum_{j=1}^{d} \kappa_j^2 \leq M\Big\}\, .
$$
Then we can say that
$$
\sum_{\nu \in \mathscr{O}_n(d,L^+(d,N)) \setminus \mathcal{S}_M} |c_r(\nu)|^2\, \leq\, \frac{\gamma \cdot [L^+(d,N)]^{-2} (m+\epsilon) C_3(d,n)}{2[L^+(d,N)]^{-2}M}\,
=\, \frac{\gamma (m+\epsilon) C_3(d,n)}{2M}\, .
$$
Let us define the truncated function, where we cut-off the coefficients in energy modes higher than the natural cutoff related to $M$:
$$
\widetilde{G}_{r,M}\, =\, \sum_{\nu \in \mathcal{S}_M} c_r(\nu) \widetilde{F}^{(n)}_{d,L^+(d,N)}(\nu;\cdot)\, .
$$
Then we have shown that $\|G_r - \widetilde{G}_{r,M}\|^2 \leq C'_3(d,n) (m+\epsilon)/M$ for some constant $C_3'(d,n)$.
Moreover, by (\ref{eq:contractionTildeT}), again, this implies
$$
\|\Phi_r - \widetilde{T}^{(n)}_{d,N} \widetilde{G}_{r,M}\|^2\, =\, \|\widetilde{T}^{(n)}_{d,N}(G_r - \widetilde{G}_{r,M})\|^2\, \leq\, 
\|G_r - \widetilde{G}_{r,M}\|^2\, \leq\, \frac{ C'_3(d,n) (m+\epsilon)}{M}\, .
$$
So we first choose $M$ to be a sufficiently large, but fixed, multiple of $m$ so that this fraction is small.
In particular, then we may assume that the vectors $\widetilde{T}^{(n)}_{d,N} \widetilde{G}_{r,M}$ are approximately orthonormal, where
each inner product deviates from $\delta_{r,s}$ by an amount not larger than $O((m+\epsilon)/M)$.
Then we notice that $\mathcal{S}_M$ is still finite, and in particular the energy of $\widetilde{F}^{(n)}_{d,L^+(d,N)}(\nu;\cdot)$ for any $\nu \in \mathcal{S}_M$
is actually bounded by $\gamma \cdot [L^+(d,N)]^{-2} M$.
So, if we assume that $N$ is large, then $\widetilde{T}^{(n)}_{d,N}$ does not deform each $\widetilde{G}_r$ by very much.
Another way of saying this is as follows.
We have
$$
\widetilde{T}^{(n)}_{d,N} \widetilde{G}_r\,
=\, \sum_{\nu \in \mathcal{S}_M} c_r(\nu) \widetilde{\Psi}_{d,N}^{(n)}(\k_1(\nu),\dots,\k_n(\nu))\, ,
$$
where we pick one element $(\k_1(\nu),\dots,\k_n(\nu)) \in \mathfrak{K}_{d,L^+(d,N)}(\nu)$ for each $\nu \in \mathcal{S}_M$.

Now, using the fact that $\Phi_r$ is an exact eigenvector, we have, for each $\nu \in \mathcal{S}_M$,
\begin{align*}
\lambda_r \langle \widetilde{\Psi}_{d,N}^{(n)}(\k_1(\nu),\dots,\k_n(\nu))\, ,\
\Phi_r \rangle\,
&=\, \langle \widetilde{\Psi}_{d,N}^{(n)}(\k_1(\nu),\dots,\k_n(\nu))\, ,\
\gamma^{-1} L^2 H_{\Lambda(d,N)}\Phi_r \rangle \\
&=\, \langle \gamma^{-1} L^2 H_{\Lambda(d,N)} \widetilde{\Psi}_{d,N}^{(n)}(\k_1(\nu),\dots,\k_n(\nu))\, ,\
\Phi_r \rangle \\
&=\, \|\nu\|^2 \cdot \langle \widetilde{\Psi}_{d,N}^{(n)}(\k_1(\nu),\dots,\k_n(\nu))\, ,\
\Phi_r \rangle \\
&\qquad + \langle (\gamma^{-1} L^2 H_{\Lambda(d,N)} -\|\nu\|^2) \widetilde{\Psi}_{d,N}^{(n)}(\k_1(\nu),\dots,\k_n(\nu))\, ,\
\Phi_r \rangle \, .
\end{align*}
So we see that , assuming $\|\nu\|^2 \neq \lambda_r$, we have
\begin{align*}
|\langle \widetilde{\Psi}_{d,N}^{(n)}(\k_1(\nu),\dots,\k_n(\nu))\, ,\
\Phi_r \rangle|^2\,
&\leq\, \frac{\left|
\langle (\gamma^{-1} L^2 H_{\Lambda(d,N)} -\|\nu\|^2) \widetilde{\Psi}_{d,N}^{(n)}(\k_1(\nu),\dots,\k_n(\nu))\, ,\
\Phi_r \rangle \right|^2}{|\lambda_r - \|\nu\|^2|^2}\\
&\leq\, \frac{\|(\gamma^{-1} L^2 H_{\Lambda(d,N)} -\|\nu\|^2) \widetilde{\Psi}_{d,N}^{(n)}(\k_1(\nu),\dots,\k_n(\nu))\|^2}
{|\lambda_r - \|\nu\|^2|^2}\, ,
\end{align*}
by Cauchy-Schwarz.
But we have already established in 
Proposition \ref{prop:Variational1} that the numerator is vanishingly small as $L \to \infty$.
Therefore, since $\lambda_r < m+\epsilon$, if $\nu \in \mathcal{S}_M \setminus \mathcal{S}_m$, then we can establish that
$|\langle \widetilde{\Psi}_{d,N}^{(n)}(\k_1(\nu),\dots,\k_n(\nu))\, ,\
\Phi_r \rangle|^2$ is small.

But since $\Phi_r$ is close to $\widetilde{T}^{(n)}_{d,N} \widetilde{G}_r$ in norm, this implies that $c_r(\nu)$ is small for such $\nu$'s.
So, defining
$$
\widetilde{\Phi}_{r,m}\, =\, \sum_{\nu \in \mathcal{S}_m} c_r(\nu) \widetilde{\Psi}_{d,N}^{(n)}(\k_1(\nu),\dots,\k_n(\nu))\, ,
$$
we have that $\|\Phi_r - \widetilde{\Phi}_r\|$ is small for each $r$.
But $\Phi_1,\dots,\Phi_R$ are orthonormal, while each $\widetilde{\Phi}_r$ is in a subspace spanned by $|\mathcal{S}_m|$ vectors.
Therefore, this requires $R\leq \mathcal{S}_m$.
Combined with Lemma \ref{lem:UpperBdVar}, this establishes the correct dimensions for the spectral subspaces.
By Proposition \ref{prop:Variational1}, we have the correct number of trial wavefunctions which are also approximate eigenvectors.
Once this is established, the remainder of the proof immediately follows.
\end{proofof}

\section{Proof of Key Step II -- Part C: Trace theorem type bound}
\label{sec:trace0}

In this section we prove
Proposition \ref{prop:trace0}.
We consider $F \in \mathfrak{S}_{\B^d(L^+(d,N))}^{(n)}(\ell^2(\B^d(L^+(d,N))^n))$.
Let us allow ourselves to abbreviate $L^+(d,N)$ as $L$.
Then
$$
\|F\|^2 - \|\widetilde{T}^{(n)}_{d,N} F\|^2\, 
=\, \sum_{(\r_1,\dots,\r_n) \in \B^d(L)^n} (1 - \mathbf{1}_{\Lambda(d,N)^n \setminus \mathscr{I}_n(\Lambda(d,N))}(\r_1,\dots,\r_n)) |F(\r_1,\dots,\r_n)|^2\, .
$$
Note that the first factor in the summation selects only those points
$(\r_1,\dots,\r_n)\in\B^d(L)^n$  in the complement of $\Lambda(d,N)^n \setminus \mathscr{I}_n(\Lambda(d,N))$.
This set is the union of $\B^d(L)^n \setminus \Lambda(d,N)^n$
and $\mathscr{I}_n(\Lambda(d,N))$.
So
$$
\|F\|^2 - \|\widetilde{T}^{(n)}_{d,N} F\|^2\, 
=\, \sum_{(\r_1,\dots,\r_n) \in \B^d(L)^n \setminus \Lambda(d,N)^n} |F(\r_1,\dots,\r_n)|^2
+ \sum_{(\r_1,\dots,\r_n) \in \mathscr{I}_n(\Lambda(d,N))} |F(\r_1,\dots,\r_n)|^2\, .
$$
We will consider these two sums, separately.
But both will use the same elementary inequality which is a discrete version of the trace theorem inequality
(as it is called in the theory of PDE's):
\begin{lemma}
\label{lem:trace1}
Suppose that $f : \{1,\dots,L\} \to \C$ is a function. 
Then
\begin{equation}
\label{eq:trace1}
|f(1)|^2\, \leq\, \frac{2L}{3}\, \sum_{\ell=1}^{L-1} |f(\ell)-f(\ell+1)|^2 + \frac{2(L-1)}{L^2}\, \sum_{\ell=1}^{L} |f(\ell)|^2\, .
\end{equation}
\end{lemma}
\begin{proof}
Let $\phi(\ell) = (L-\ell)/(L-1)$.
Then, we have
$$
f(1)\, =\, 
f(1)\phi(1) - f(L) \phi(L)\,
=\, \sum_{\ell=1}^{L-1} [f(\ell)\phi(\ell) - f(\ell+1)\phi(\ell+1)]\, .
$$
Then we may rewrite
\begin{align*}
f(\ell) \phi(\ell) - f(\ell+1)\phi(\ell+1)\,
&=\, \phi(\ell) [f(\ell)-f(\ell+1)] + [\phi(\ell)-\phi(\ell+1)] f(\ell+1)\\ 
&=\, \phi(\ell) [f(\ell)-f(\ell+1)] + \frac{1}{L-1}\, f(\ell+1)\, .
\end{align*}
So
$$
f(1)\, =\, \sum_{\ell=1}^{L-1} \phi(\ell) [f(\ell)-f(\ell+1)]  + \frac{1}{L-1}\, \sum_{\ell=1}^{L-1} f(\ell+1)\, .
$$
Note that we can make the second sum slightly more symmetric, by including $f(1)$ to obtain
$$
\frac{L}{L-1}\, f(1)\, =\, 
\sum_{\ell=1}^{L-1} \phi(\ell) [f(\ell)-f(\ell+1)]  + \frac{1}{L-1}\, \sum_{\ell=1}^{L} f(\ell)\, . 
$$
Then, by the triangle inequality and Cauchy-Schwarz,
$$
\frac{L}{L-1}\, |f(1)|\, \leq\, \bigg[\sum_{\ell=1}^{L-1} \phi(\ell)^2\bigg]^{1/2} \bigg[\sum_{\ell=1}^{L-1} |f(\ell)-f(\ell+1)|^2\bigg]^{1/2}
+ \frac{1}{\sqrt{L-1}}\, \bigg[\sum_{\ell=1}^{L-1} |f(\ell)|^2\bigg]^{1/2}\, .
$$
Note that $\sum_{\ell=1}^{L-1} \phi(\ell)^2 = (L-1)^{-2} \sum_{\ell=1}^{L-1} \ell^2$ and performing the sum, this is bounded by $L^3/[3(L-1)^2]$.
Squaring and using the inequality $(a+b)^2 \leq 2a^2+2b^2$ gives the result.
\end{proof}
Now we return to the problem of bounding
$$
\sum_{(\r_1,\dots,\r_n) \in \B^d(L)^n \setminus \Lambda(d,N)^n} |F(\r_1,\dots,\r_n)|^2\, ,\quad
\text{ and }\quad
\sum_{(\r_1,\dots,\r_n) \in \mathscr{I}_n(\Lambda(d,N))} |F(\r_1,\dots,\r_n)|^2\, .
$$
Given a point $(\r_1,\dots,\r_n) \in \B^d(L)^n \setminus \Lambda(d,N)^n$, there is at least one $k_0$ such that 
$\r_{k_0} \in \B^d(L) \setminus \Lambda(d,N)$.
(We choose $k_0$ to be the minimal element in case the number of possible choices is greater than 1.)
Then, there is also some $j_0 \in \{1,\dots,d\}$ such that $r_{k_0,j_0}=L$. (Again, in case of multiple choices we choose the minimal one.)
Now we will consider a chain of points $(\widetilde{\r}_1(\ell),\dots,\widetilde{\r}_n(\ell))$ for $\ell \in \{0,\dots,\lfloor L/2\rfloor -1\}$,
as follows.
We let $\widetilde{r}_{k_0,j_0}(\ell) = L-\ell$
and $\widetilde{r}_{k,j}(\ell) = r_{k,j}$ for all $(k,j) \neq (k_0,j_0)$.
Let us actually refer to these points as $(\widetilde{\r}_1(\ell;\r_1,\dots,\r_n),\dots,\widetilde{\r}_n(\ell;\r_1,\dots,\r_n))$, for later reference.
By Lemma \ref{lem:trace1}, we then know
\begin{equation}
\label{eq:accretion}
\begin{split}
&\hspace{-0.25cm}\sum_{(\r_1,\dots,\r_n) \in \B^d(L)^n \setminus \Lambda(d,N)^n} |F(\r_1,\dots,\r_n)|^2\\ 
&\hspace{1cm}\leq\, \frac{L}{3} \sum_{(\r_1,\dots,\r_n) \in \B^d(L)^n \setminus \Lambda(d,N)^n} \sum_{\ell=1}^{\lfloor L/2\rfloor-1} 
|F(\widetilde{\r}_1(\ell;\r_1,\dots,\r_n))-
F(\widetilde{\r}_1(\ell-1;\r_1,\dots,\r_n))|^2\\
&\hspace{1cm} + \frac{2}{\lfloor L/2 \rfloor}\, 
\sum_{(\r_1,\dots,\r_n) \in \B^d(L)^n \setminus \Lambda(d,N)^n} \sum_{\ell=0}^{\lfloor L/2\rfloor-1} 
|F(\widetilde{\r}_1(\ell;\r_1,\dots,\r_n))|^2
\end{split}
\end{equation}
Now, let us count the number of times any given edge in $\Phi_n(\B^d(L))$, say going from $(\r_1',\dots,\r_n')$ to $(\r_1'',\dots,\r_n'')$, will occur as an edge,
from $(\widetilde{\r}_1(\ell;\r_1,\dots,\r_n))$ to $(\widetilde{\r}_1(\ell-1;\r_1,\dots,\r_n))$,
for some point $(\r_1,\dots,\r_n) \in \B^d(L)^n \setminus \Lambda(d,N)^n$.
Note that $r_{kj}' = r''_{kj}$ except for one choice $(k_0,j_0)$.
But then we may determine that we must have $(\r_1,\dots,\r_n)$ being the point such that $r_{k_0,j_0}=L$ and $r_{kj}=r'_{kj}$ for all $(k,j) \neq (k_0,j_0)$.
So it can only occur at most one time.
Thus, since each $(\r_1',\dots,\r_n')$ is an endpoint for $2dn$ edges (at most), we deduce that
\begin{equation}
\label{eq:firstSUM}
\sum_{(\r_1,\dots,\r_n) \in \B^d(L)^n \setminus \Lambda(d,N)^n} |F(\r_1,\dots,\r_n)|^2\,
\leq\, \frac{2L}{3}\, \langle F, \mathfrak{H}^{(n)}_{\B^d(L)} F\rangle + \frac{4nd}{\lfloor L/2 \rfloor}\, \|F\|^2\, .
\end{equation}
Let us choose this point to declare that we are proving the proposition

\begin{proofof}{\bf Proof of Proposition \ref{prop:trace0}:}
All that remains is to bound the second sum
$$
\sum_{(\r_1,\dots,\r_n) \in \mathscr{I}_n(\Lambda(d,N))} |F(\r_1,\dots,\r_n)|^2\, .
$$
We use a similar argument as before. If $(\r_1,\dots,\r_n)$ is in $\mathscr{I}_n(\Lambda(d,N))$, then 
we must have some $1\leq k_1<k_2\leq n$ such that $\r_{k_1} = \r_{k_2}$.
Now we will make $(\widetilde{\r}_1(\ell),\dots,\widetilde{\r}_n(\ell))$ for $\ell \in \{0,\dots,\lfloor L/2\rfloor -1\}$,
as follows.
If $r_{k_1,1} \in \{1,\dots,\lfloor L/2 \rfloor\}$ then we take $\widetilde{r}_{k_1,1}(\ell) = r_{k_1,1}+\ell$,
and we let $\widetilde{r}_{k,j}(\ell)=r_{k,j}$ for all $(k,j) \neq (k_1,\ell)$.
Otherwise, we have $r_{k_1,1} \in \{\lfloor L/2\rfloor+1,\dots,L\}$ and we let 
$\widetilde{r}_{k_1,1}(\ell) = r_{k_1,1}-\ell$.
We again have the same type of bound as in (\ref{eq:accretion}), except now we have $\mathscr{I}_n(\Lambda(d,N))$,
instead of $\B^d(L)^n \setminus \Lambda(d,N)^n$.

We must again ask, for an arbitrary edge, 
going from $(\r_1',\dots,\r_n')$ to $(\r_1'',\dots,\r_n'')$, how many times will it occur as an edge,
from $(\widetilde{\r}_1(\ell;\r_1,\dots,\r_n))$ to $(\widetilde{\r}_1(\ell-1;\r_1,\dots,\r_n))$,
for some point $(\r_1,\dots,\r_n) \in \mathscr{I}_n(\Lambda(d,N))$.
Note that there is exactly one $k_1 \in \{1,\dots,n\}$ such that $\widetilde{\r}_{k_1}' \neq \widetilde{\r}_{k_1}''$.
There is also exactly one $j \in \{1,\dots,d\}$ such that $r_{k_1,j}' \neq r_{k_1,j}''$.
If $j\neq 1$, then the number of times the edge occurs is $0$.
But if $j=1$ then we may do the following. 
Consider $\r'''(\ell) = (\ell,\r_{k_1,2},\dots,\r_{k_1,n})$.
Whenever we obtain that $\r'''(\ell)=\r_k'$ for some $k\neq k_1$ and $\ell \in \{1,\dots,L\}$
then we may consider the $n$-tuple where, starting with $(\r_1',\dots,\r_n')$, we replace $\r_{k_1}'$ by $\r'''(\ell)$.
Then this is a possible starting point $(\r_1,\dots,\r_n)$.
But these are the only possible starting points, and there are only $n-1$ possibilities because we must have $k \in \{1,\dots,n\} \setminus \{k_1\}$.

Therefore, as before, we obtain
\begin{equation*}
\sum_{(\r_1,\dots,\r_n) \in \mathscr{I}_n(\Lambda(d,N))} |F(\r_1,\dots,\r_n)|^2\,
\leq\, \frac{2(n-1)L}{3}\, \langle F, \mathfrak{H}^{(n)}_{\B^d(L)} F\rangle + \frac{4n(n-1)d}{\lfloor L/2 \rfloor}\, \|F\|^2\, .
\end{equation*}
Putting this together with (\ref{eq:firstSUM}) gives the desired result.
\end{proofof}

\section{Proof of Key Step II -- Part D: Extension theorem type bound}
\label{sec:extension0}

We will now prove Proposition \ref{prop:extension0}.
Suppose that we have $F \in \ell^2(\B^d(L)^n)$
which is supported on $\Lambda(d,N)^n \setminus \mathscr{I}_n(\Lambda(d,N))$.
Let us define
\begin{multline}
\widetilde{\mathfrak{H}}^{(n)}_{\Lambda(d,N)} F(\r_1,\dots,\r_n)\\
=\, \sum_{(\r_1',\dots,\r_n')\Lambda(d,N)^n \setminus \mathscr{I}_n(\Lambda(d,N))}
\mathbf{1}_{\{1\}}(\|(\r_1,\dots,\r_n) - (\r'_1,\dots,\r'_n)\|)
[F(\r_1,\dots,\r_n) - F(\r'_1,\dots,\r'_n)]\, .
\end{multline}
Then the goal is to prove that there is an extension $\widetilde{F}$ of $F$ to all of $\B^d(L)^d$,
such that 
\begin{itemize}
\item[(i)] $\widetilde{F}(\r_1,\dots,\r_n) = F(\r_1,\dots,\r_n)$
for each $(\r_1,\dots,\r_n) \in \Lambda(d,N)^n \setminus \mathscr{I}_n(\Lambda(d,N))$,
\item[(ii)]
there exists some fixed $C_3(d,n)$ (not depending on $F$) such that
\begin{equation}
\label{eq:CONDii}
\langle \widetilde{F},\mathfrak{H}^{(n)}_{\B^d(L)} \widetilde{F}\rangle\,
\leq\, C_3(d,n) \langle F,\widetilde{\mathfrak{H}}^{(n)}_{\Lambda(d,N)} F\rangle\, .
\end{equation}
\end{itemize}
Given any $(\r_1,\dots,\r_n) \not\in \Lambda(d,N)^n \setminus \mathscr{I}_n(\Lambda(d,N))$,
let $\rho(\r_1,\dots,\r_n)$ be the minimum distance to $\Lambda(d,N)^n \setminus \mathscr{I}_n(\Lambda(d,N))$:
$$
\rho(\r_1,\dots,\r_n)\, \stackrel{\mathrm{def}}{:=}\,
\min\{\|(\r_1,\dots,\r_n) - (\r_1',\dots,\r_n')\|_1\, :\, (\r_1',\dots,\r_n') \in \Lambda(d,N)^n \setminus \mathscr{I}_n(\Lambda(d,N))\}\, .
$$
Then, given $(\r_1,\dots,\r_n) \in \B^d(L)^n$ satisfying 
$(\r_1,\dots,\r_n) \not\in \Lambda(d,N)^n \setminus \mathscr{I}_n(\Lambda(d,N))$,
we define $(\widetilde{\r}_1,\dots,\widetilde{\r}_n) \in \Lambda(d,N)^n \setminus \mathscr{I}_n(\Lambda(d,N))$
to be the point such that $\|(\r_1,\dots,\r_n)-(\widetilde{\r}_1,\dots,\widetilde{\r}_n)\|_1 = \rho(\r_1,\dots,\r_n)$ 
and such that $(\widetilde{\r}_1,\dots,\widetilde{\r}_n)$ is minimal in the lexicographic
ordering in case more than 1 such point exists.
For later reference, let us denote these points as $(\widetilde{\r}_1(\r_1,\dots,\r_n),\dots,\widetilde{\r}_n(\r_1,\dots,\r_n))$.
Note that 
$$
(\r_1,\dots,\r_n) \in \Lambda(d,N)^n \setminus \mathscr{I}_n(\Lambda(d,N))\quad \Leftrightarrow\quad
(\widetilde{\r}_1(\r_1,\dots,\r_n),\dots,\widetilde{\r}_n(\r_1,\dots,\r_n))\, =\, (\r_1,\dots,\r_n)\, .
$$
Then we define, for every $(\r_1,\dots,\r_n) \in \B^d(L)$,
\begin{equation}
\label{eq:widetTildeFDef}
\widetilde{F}(\r_1,\dots,\r_n)\, \stackrel{\mathrm{def}}{:=}\, 
F(\widetilde{r}_1(\r_1,\dots,\r_n),\dots,\widetilde{r}_n(\r_1,\dots,\r_n))\, .
\end{equation}
This does satisfy $\widetilde{F}(\r_1,\dots,\r_n)=F(\r_1,\dots,\r_n)$
if $(\r_1,\dots,\r_n)$ is in $\Lambda(d,N)^n \setminus \mathscr{I}_n(\Lambda(d,N))$, condition (i).
So we just have to check condition (ii), namely equation (\ref{eq:CONDii}).

Now, the key point is the following:
suppose that $(\r_1,\dots,\r_n)$ and $(\r_1',\dots,\r_n')$ are neighbors in $\B^d(L)^n$,
then we have
\begin{multline}
\label{eq:diffFirst}
\widetilde{F}(\r_1,\dots,\r_n)
-\widetilde{F}(\r_1',\dots,\r_n')\\
=\, F(\widetilde{\r}_1(\r_1,\dots,\r_n),\dots,\widetilde{\r}_n(\r_1,\dots,\r_n))
- F(\widetilde{\r}_1(\r_1',\dots,\r_n'),\dots,\widetilde{\r}_n(\r_1',\dots,\r_n'))\, .
\end{multline}
The question is how to bound the right hand side, using $\langle F, \widetilde{H}^{(n)}_{\Lambda(d,N)} F\rangle$.
We are somewhat motivated by the argument of 
Section \ref{sec:trace0}.
We will construct a chain connecting
$(\widetilde{\r}_1(\r_1,\dots,\r_n),\dots,\widetilde{\r}_n(\r_1,\dots,\r_n))$
to $(\widetilde{\r}_1(\r_1',\dots,\r_n'),\dots,\widetilde{\r}_n(\r_1',\dots,\r_n'))$ within
$\Lambda(d,N)^n \setminus \mathscr{I}_n(\Lambda(d,N))$
and then we will use a telescoping sum.

We combine these two lemmas:
\begin{lemma}
\label{lem:MAXrho}
There is a constant $\rho_{\max}(n,d)$ such that 
for any $(\r_1,\dots,\r_n) \in \B^d(L)$, we have 
$$
\rho(\r_1,\dots,\r_n)\, \leq\, \rho_{\max}(n,d)\, .
$$
\end{lemma}
\begin{lemma}
\label{lem:LITTLEc3}
There is a constant $c_3(d,n)$ such that the following holds.
If $(\r_1,\dots,\r_n)$ and $(\r_1',\dots,\r_n')$ are two points in $\Lambda(d,N)^n \setminus \mathscr{I}_n(\Lambda(d,N))$ satisfying
$$
\|(\r_1,\dots,\r_n) - (\r_1',\dots,\r_n')\|\, \leq\, 2 \rho_{\max}(n,d)+1\, ,
$$
then, for some $\tau \leq c_3(d,n)$, there are points $(\hat{\r}_1(t),\dots,\hat{\r}_n(t)) \in \Lambda(d,N)^n \setminus \mathscr{I}_n(\Lambda(d,N))$, 
for $t \in \{0,\dots,\tau\}$,
satisfying
\begin{itemize}
\item $(\hat{\r}_1(0),\dots,\hat{\r}_n(0)) = (\r_1,\dots,\r_n)$,
\item $(\hat{\r}_1(\tau),\dots,\hat{\r}_2(\tau)) = (\r_{\pi_1},\dots,\r_{\pi_n})$, for some permutation $\pi \in S_n$,
where $\pi$ equals the identity if $d>1$,
and 
\item
for each $t \in \{1,\dots,\tau\}$,
$$
\|(\hat{\r}_1(t),\dots,\hat{\r}_2(t)) - (\hat{\r}_1(t-1),\dots,\hat{\r}_2(t-1))\|_1\, =\, 1\, .
$$
\end{itemize}
\end{lemma}

These lemmas are not difficult. But they require detailed descriptions.
We wish to state their implication, first, which is the proof of  Proposition \ref{prop:extension0}.

\begin{proofof}{\bf Proof of Proposition \ref{prop:extension0}:}
Suppose that $(\r_1,\dots,\r_n)$ and $(\r_1',\dots,\r_n')$ are two points in $\B^d(L)^n$ 
satisfying
$$
\|(\r_1,\dots,\r_n) - (\r_1',\dots,\r_n')\|_1\, \leq\, 1\, .
$$
Then by Lemma \ref{lem:MAXrho} and the triangle inequality, we see that
$$
\|(\widetilde{\r}_1(\r_1,\dots,\r_n),\dots,\widetilde{\r}_n(\r_1,\dots,\r_n))
- (\widetilde{\r}_1(\r_1',\dots,\r_n'),\dots,\widetilde{\r}_n(\r_1',\dots,\r_n'))\|_1\, \leq 2\rho_{\max}(n,d)+1\, .
$$
Then, by Lemma \ref{lem:LITTLEc3}, there is a chain
$(\hat{\r}_1(t),\dots,\hat{\r}_n(t)) \in \Lambda(d,N)^n \setminus \mathscr{I}_n(\Lambda(d,N))$,
for $t \in \{0,\dots,\tau\}$
linking $(\widetilde{\r}_1(\r_1,\dots,\r_n),\dots,\widetilde{\r}_n(\r_1,\dots,\r_n))$ and 
$(\widetilde{\r}_{\pi_1}(\r_1',\dots,\r_n'),\dots,\widetilde{\r}_{\pi_n}(\r_1',\dots,\r_n'))$, for some $\pi \in S_n$.
The fact that there is a permutation, does not concern us because $F$ is a symmetric function.
In other words, we then have, by (\ref{eq:diffFirst}) a telescoping sum
$$
\widetilde{F}(\r_1,\dots,\r_n)
-\widetilde{F}(\r_1',\dots,\r_n')\,
\leq\,
\sum_{t=1}^{\tau} F((\hat{\r}_1(t),\dots,\hat{\r}_n(t)))
- F((\hat{\r}_1(t-1),\dots,\hat{\r}_n(t-1)))\, .
$$
Then, by the Cauchy-Schwarz inequality, we obtain
$$
|\widetilde{F}(\r_1,\dots,\r_n)
-\widetilde{F}(\r_1',\dots,\r_n')|^2\,
\leq\,
\tau \sum_{t=1}^{\tau} |F((\hat{\r}_1(t),\dots,\hat{\r}_n(t)))
- F((\hat{\r}_1(t-1),\dots,\hat{\r}_n(t-1)))|^2\, .
$$
But by Lemma \ref{lem:LITTLEc3}, we know that $\tau\leq c_3(d,n)$.
So we obtain
\begin{equation}
\label{eq:chainDIFF}
|\widetilde{F}(\r_1,\dots,\r_n)
-\widetilde{F}(\r_1',\dots,\r_n')|^2\,
\leq\,
c_3(d,n) \sum_{t=1}^{\tau} |F((\hat{\r}_1(t),\dots,\hat{\r}_n(t)))
- F((\hat{\r}_1(t-1),\dots,\hat{\r}_n(t-1)))|^2\, .
\end{equation}
Now let us begin to enumerate how many times a given edge will be chosen in some term such as the right hand side above, when we apply this bound
to every summand of the formula
\begin{equation}
\label{eq:ENTOT}
\langle \widetilde{F}, \mathfrak{H}^{(n)}_{\B^d(L)} \widetilde{F}\rangle\,
=\, \frac{1}{2}\, \sum_{
\{(\r_1,\dots,\r_n),
(\r_1',\dots,\r_n')\} \in \Phi_n(\B^d(L))}
|\widetilde{F}(\r_1,\dots,\r_n)
-\widetilde{F}(\r_1',\dots,\r_n')|^2\, ,
\end{equation}
We do not seek a reasonable bound, just a finite one.
Therefore, given any pair 
$$
(\r_1^{(1)},\dots,\r_n^{(1)}),
(\r_1^{(2)},\dots,\r_n^{(2)}) \in \Lambda(d,N)^n \setminus \mathscr{I}_n(\Lambda(d,N))\, ,
$$
of distance 1 apart, we see that if it is in a chain then there are at most 
$$
\mathfrak{N}(d,n)\, \stackrel{\mathrm{def}}{:=}\, \sum_{\tau=2}^{c_3(d,n)}(\tau-1)(nd)^{\tau-1}\, ,
$$ 
choices for how
the chain continues to the left and the right (since at each step, the chain has at most $(nd)$ choices for how to take its next step,
and the length of the left and right sides of the chain emanating from our initial edge must be $(k,\tau-1-k)$ for some $k \in \{0,\dots,\tau-1\}$).
Given $((\hat{\r}_1(0),\dots,\hat{\r}_n(0)))$
and $((\hat{\r}_1(\tau),\dots,\hat{\r}_n(\tau)))$, there is still a choice of the two points $(\r_1,\dots,\r_n),(\r_1',\dots,\r_n') \in \B^d(L)^n$
such that
$$
(\widetilde{\r}_1(\r_1,\dots,\r_n),\dots,\widetilde{\r}_n(\r_1,\dots,\r_n))\, =\, ((\hat{\r}_1(0),\dots,\hat{\r}_n(0)))
$$
and
$$
(\widetilde{\r}_1(\r_1',\dots,\r_n'),\dots,\widetilde{\r}_n(\r_1',\dots,\r_n'))\, =\, ((\hat{\r}_1(\tau),\dots,\hat{\r}_n(\tau)))\, .
$$
But, no matter what value of $\tau$ we have, we may just use Lemma \ref{lem:MAXrho}
to say that $\|(\r_1,\dots,\r_n)-((\hat{\r}_1(0),\dots,\hat{\r}_n(0)))\|_1$ 
and $\|(\r_1',\dots,\r_n') - ((\hat{\r}_1(\tau),\dots,\hat{\r}_n(\tau)))\|_1$
are bounded by $\rho_{\max}(n,d)$.
Since  $((\hat{\r}_1(0),\dots,\hat{\r}_n(0)))$
and $((\hat{\r}_1(\tau),\dots,\hat{\r}_n(\tau)))$ have been already chosen,
this implies that there are at most $(2\rho_{\max}(d,n)+1)^{nd}$ choices for each of the points.
But there is also a possibility of a permutation $\pi \in S_n$.
So, we have determined that the pair of points
$
(\r_1^{(1)},\dots,\r_n^{(1)}),
(\r_1^{(2)},\dots,\r_n^{(2)}) \in \Lambda(d,N)^n \setminus \mathscr{I}_n(\Lambda(d,N))$
can come from at most 
$$
\mathfrak{N}'(d,n)\, 
\stackrel{\mathrm{def}}{:=}\, (2\rho_{\max}(d,n)+1)^{2nd} (n!) \mathfrak{N}(d,n)
$$
terms in the right hand side of (\ref{eq:chainDIFF}) when we use these to bound all the summands in (\ref{eq:ENTOT}).
Thus, combining all this, we obtain
$$
\langle \widetilde{F}, \mathfrak{H}^{(n)}_{\B^d(L)} \widetilde{F}\rangle\,
\leq\, \frac{1}{2}\, c_3(d,n) \mathfrak{N}'(d,n) \langle F, \widetilde{H}^{(n)}_{\Lambda(d,N)} F\rangle\, .
$$
\end{proofof}

One clear fact is that our bounds are quite large. This will also be the case in the proof of Lemma \ref{lem:MAXrho}
and Lemma \ref{lem:LITTLEc3}.

\begin{proofof}{\bf Proof of Lemma \ref{lem:MAXrho}:}
Let us write $\bR$ for $(\r_1,\dots,\r_n)$, in order to simplify notation.
Given $\bR = (\r_1,\dots,\r_n) \in \B^d(L)^n$ we will
find a $\tau$ and a sequence $\check{\bR}(t) = (\check{\r}_1(t),\dots,\check{\r}_n(t)) \in \B^d(L)^n$ for $t \in \{0,\dots,\tau\}$,
which is a chain (so that $\|\check{\bR}(t)-\check{\bR}(t-1)\|_1=1$ for each $t \in \{1,\dots,\tau\}$)
such that $\check{\bR}(0) = \bR$ and $\check{\bR}(\tau) \in \Lambda(d,N)$.
Then if we find a uniform bound on $\tau$, that will give $\rho_{\max}(d,n)$.

Given $\bR$, we will call this a point-vector. We will call each $\r_1,\dots,\r_k$ the points. Sometimes we just refer to $k$ as point index.
We refer to $j \in \{1,\dots,d\}$ as a coordinate index, since those index the coordinates $\r_k = (r_{k,1},\dots,r_{k,d})$ for some point $\r_k$.

\smallskip
\noindent
\underline{Step 1:}
First let $\tau_1=\sum_{k=1}^{n} \sum_{j=1}^{d} \mathbf{1}_{\{L\}}(r_{k,j})$.
We want to push all coordinates of all points back from $L$, so that at the end no point of $\check{\bR}(\tau_1)$ is in $\B^d(L)$.
I.e., we enumerate those pairs $(k,j)$ such that $r_{k,j}=L$.
At each $t \in \{1,\dots,\tau_1\}$, we choose the $t$th pair $(k,j)$ and we let $r_{k,j}(t)=L-1$, whereas, we had $r_{k,j}(t-1)=L$.

Note that we can only move 1 coordinate of 1 point at each step. So it suffices to say which one of these changes and how it changes (by going up or down by 1 step).

So, in the end, $\check{\bR}(\tau_1)$ has $\check{r}_{k,j}(\tau_1) = \min\{L-1,r_{k,j}\}$ for each $(k,j)$.
Note, $\tau_1\leq nd$.

\smallskip
\noindent
\underline{Step 2:}
This is the last step for this proof. It is harder. It involves a procedure we call ``clustering'' and ``spreading apart.''

Let us define $\approx$ to be a symmetric relation on $\{1,\dots,n\}$ (although not transitive)
wherein $k_1 \approx_n k_2$ if $\|\check{\r}_{k_1}(\tau)-\check{\r}_{k_2}\|_1<2n-1$.
The reason for the choice of the distance, $2n-1$, will become more apparent, later.
Then we define an equivalence relation $\sim$ wherein $k_1 \sim k_2$ if and only if 
there is some $m$ and some sequence $\widehat{k}_0,\dots,\widehat{k}_m$ such that 
$\widehat{k}_0=k_1$, $\widehat{k}_m=k_2$ and $\widehat{k}_{p} \approx_n \widehat{k}_{p-1}$
for each $p \in \{1,\dots,m\}$.
We define the equivalence classes of $\sim$ to be the ``clusters.''

Now we will only focus on the first coordinate $j=1$, in this proof.
(But in the proof of Lemma \ref{lem:LITTLEc3}, we will need an algorithm similar to this
applied to all the coordinates.)
Our goal is to find $\tau_2$ and to extend $\check{\bR}(\tau_1+t)$ for $t \in \{1,\dots,\tau_2\}$
so that $\check{r}_{1,1}(\tau_1+\tau_2),\check{r}_{2,1}(\tau_1+\tau_2),\dots,\check{r}_{n,1}(\tau_1+\tau_2)$ are all distinct.
I.e., we want the points of $\check{\bR}(\tau_1+\tau_2)$ to have distinct $1$st coordinates.

Suppose $\mathcal{K}_1,\dots,\mathcal{K}_m$ are the clusters of $\{1,\dots,n\}$.
For each $p \in \{1,\dots,m\}$, we enumerate the cluster $\mathcal{K}_p$ as $\{k_{p,1},\dots,k_{p,|\mathcal{K}_p|}\}$ for some $r(p)$, where we choose the ordering such that 
$$
\check{r}_{k_{p,r},1}(\tau_1)\, \leq\,\check{r}_{k_{p,s},1}(\tau_1)\, ,
$$
for each pair $r,s$ with $1\leq r<s\leq |\mathcal{K}_p|$. 
(In case the 1st coordinate of some point indices in $\mathcal{K}_p$ coincide, 
we may choose the ordering for those coinciding point indices in any arbitrary way, such as the usual order.)

Now, if $\check{r}_{k_{p,1},1}(\tau_1) \in \{\lfloor L/3 \rfloor,\dots,L-1\}$, then we say that $p$ enumerates a high cluster,
or $\mathcal{K}_p$ is a high cluster.
If this fails to happen, but $\check{r}_{k_{p,|\mathcal{K}_p|},1}(\tau_1) \in \{1,\dots,L-\lfloor L/3\rfloor\}$,
then we say $p$ enumerates a low cluster.

Now, also, note that since $k_{p,1} \sim k_{p,|\mathcal{K}_p|}$, this means that 
$$
k_{p,|\mathcal{K}_p|} - k_{p,1}\, \leq\, (|\mathcal{K}_p|-1)(n-1)\, \leq\, (n-1)^2\, .
$$
Therefore, if $\mathcal{K}_p$ is not a high cluster, then $k_{p,1}<\lfloor L/3 \rfloor$,
and that means $k_{p,|\mathcal{K}_p|} \leq \lfloor L/3 \rfloor + (n-1)^2 - 1$.
So, in this case $\mathcal{K}_p$ is a low cluster, as long as we have
$$
L-\lfloor L/3 \rfloor\, \geq\, \lfloor L/3 \rfloor + (n-1)^2 - 1\qquad \Leftrightarrow\qquad 
L-2\lfloor L/3\rfloor\, \geq\, (n-1)^2-1=n(n-2)
$$
In turn, this is insured if $L\geq 3n(n-2)$.
So, henceforth, we assume $L$ satisfies this lower bound.
We also always assume $L\geq 3n$ (in case $n-2\leq 1$).

\underline{High cluster case:} Recursively, for each $p \in \{1,\dots,m\}$, we define $\tau_{2,p}$ and $\check{\bR}(\tau_1+\tau_{2,1}+\dots+\tau_{2,p-1}+t)$
for $t \in \{1,\dots,\tau_{2,p}\}$ in order to do the following.
If $\mathcal{K}_p$ is a high cluster, then we take
$$
\tau_{2,p} = \sum_{r=1}^{|\mathcal{K}_p|} \check{r}_{k_{p,r},1}(\tau_1+\tau_{2,1}+\dots+\tau_{2,p-1}) - \sum_{r=1}^{|\mathcal{K}_p|} 
\left[\check{r}_{k_{p,r},1}(\tau_1+\tau_{2,1}+\dots+\tau_{2,p-1})+r-1\right]\, .
$$
This is the minimum number of steps needed to send the coordinates $\check{r}_{k_{p,r},1}(\tau_1+\tau_{2,1}+\dots+\tau_{2,p-1})$ to 
$\left[\check{r}_{k_{p,r},1}(\tau_1+\tau_{2,1}+\dots+\tau_{2,p-1})+r-1\right]$.
That is how we prescribe $\check{\bR}(\tau_1+\tau_{2,1}+\dots+\tau_{2,p-1}+t)$
for $t \in \{1,\dots,\tau_{2,p}\}$.

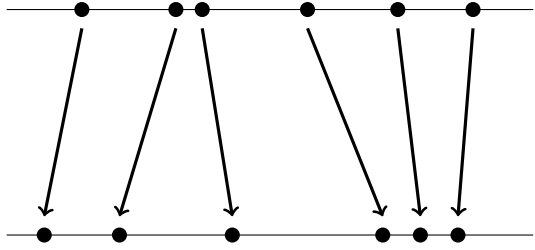
\begin{figure}
\begin{center}
\begin{tikzpicture}
\draw (0,0) -- (7,0);
\fill (1,0) circle (1mm);
\fill (2.25,0) circle (1mm);
\fill (2.6,0) circle (1mm);
\fill (4,0) circle (1mm);
\fill (5.2,0) circle (1mm);
\fill (6.2,0) circle (1mm);
\fill (0.5,-3) circle (1mm);
\fill (1.5,-3) circle (1mm);
\fill (3,-3) circle (1mm);
\fill (5,-3) circle (1mm);
\fill (5.5,-3) circle (1mm);
\fill (6,-3) circle (1mm);
\draw[very thick,->] (1,-0.25) -- (0.5,-2.75);
\draw[very thick,->] (2.25,-0.25) -- (1.5,-2.75);
\draw[very thick,->] (2.6,-0.25) -- (3,-2.75);
\draw[very thick,->] (4,-0.25) -- (5,-2.75);
\draw[very thick,->] (5.2,-0.25) -- (5.5,-2.75);
\draw[very thick,->] (6.2,-0.25) -- (6,-2.75);
\draw (0,-3) -- (7,-3);
\end{tikzpicture}
\end{center}
\caption{ Left-movers are $1$st, $2$nd, $6$th particles, and 
right-movers are the $3$rd, $4$th and $5$th.
\label{fig:leftrightmovers}}
\end{figure}
Actually, now we will be more precise about a particular order for the updates (because this will be useful in the proof of 
Lemma \ref{lem:LITTLEc3}).
Let us say that $k_{p,r}$ is a ``left-mover'' if
$$
\check{r}_{k_{p,r},1}(\tau_1+\tau_{2,1}+\dots+\tau_{2,p-1}) > \left[\check{r}_{k_{p,r},1}(\tau_1+\tau_{2,1}+\dots+\tau_{2,p-1})+r-1\right]\ ,
$$
and say it is a ``right-mover'' if the strict inequality is reversed.
We say it is stationary if there is equality instead of a strict inequality.
See Figure \ref{fig:leftrightmovers}.

We move the left-movers first, and then the right-movers, second.
Moreover, we move the left-movers, starting the the left-most one, and proceeding to the right.
Then since the final point coordinates are ordered (in the same order), we can see that no left mover collides with a right-mover which is staying still, since any right-mover to the left of the left-mover
initially, also ends up to the left of the left-mover finally, and it starts even more to the left than that (while it is staying still).
Also, no left mover collides with any right mover to its right, because the left-mover is moving right.
By the time the left-mover moves, all the left-movers to its left have moved to their final positions, which is to the left of its final positions.
So it does not collide with any of those, and it also does not collide with any of the left-movers to its right, because they stay stationary during its moves,
and they are to its right (like the right-movers to the right).
A similar argument applies to the right-movers when it is their turn to move.
(One could do reflection and time-reversal to flip the roles of left- and right-movers, to see this, if one so desired.)

This focus on non-collision will be more important during the proof of Lemma \ref{lem:LITTLEc3}.

\underline{Low cluster case:} If $\mathcal{K}_p$ is a low cluster, then it will be clear by symmetry what we do, based on the high cluster case.
We take
$$
\tau_{2,p} = -\sum_{r=1}^{|\mathcal{K}_p|} \check{r}_{k_{p,|\mathcal{K}_p|+1-r},1}(\tau_1+\tau_{2,1}+\dots+\tau_{2,p-1}) + \sum_{r=1}^{|\mathcal{K}_p|} 
\left[\check{r}_{k_{p,|\mathcal{K}_p|},1}(\tau_1+\tau_{2,1}+\dots+\tau_{2,p-1})+1-r\right]\, .
$$
This is the minimum number of steps needed to send the coordinates
$\check{r}_{k_{p,|\mathcal{K}_p|+1-r},1}(\tau_1+\tau_{2,1}+\dots+\tau_{2,p-1})$
to  
$\left[\check{r}_{k_{p,|\mathcal{K}_p|},1}(\tau_1+\tau_{2,1}+\dots+\tau_{2,p-1})+1-r\right]$
for each $r \in \{1,\dots,|\mathcal{K}_p|\}$.

We do it in the same way as before, non-colliding. We will not repeat the details, here.

\underline{Conclusion:}
We know that the 1st coordinate of all the points in a given cluster are now spread apart, so they are not intersecting.
We still need to check that points from distinct clusters did not move their first coordinates so as to now intersect.
But in the high case, we moved the points at most $|\mathcal{K}_p|-1$ points to the right (for example if they all began intersecting at the coordinate of the left-most point).
We did not move them to the left of the left-most point.
So they are within $n-1$ of the original left-most-point.
Similarly, in the low case, they are within $n-1$ of the right-most-point.

But, initially, points from distinct clusters have their first coordinates separated by at least $2n-1$.
So after the final time, we do not have any points from distinct clusters with 1st coordinate smaller than 1 apart.
Thus we have achieved the desired goal of moving all the 1st coordinates apart at time $\tau_1 + \tau_2$,
where $\tau_2 = \tau_{2,1} + \dots +\tau_{2,m}$.

Now we want to bound $\tau_2$.
Note that in each cluster, initially the diameter is at most $(n-1)(2n-2)$.
So some thought shows that each point moves at most $\max\{n-1,(n-1)(2n-2)\}$.
Since there are $n$ points to move, we can bound $\tau_2$ by $n\max\{n-1,(n-1)(2n-2)\}$.
(We assume $n>1$, otherwise there was no need for Step 2.)
Since we already had a bound on $\tau_1$, we have a bound on $\tau=\tau_1+\tau_2$.
We call this bound $\rho_{\max}(n,d)$. It is $nd + n\max\{n-1,(n-1)(2n-2)\}$.
\end{proofof}

\begin{proofof}{\bf Proof of Lemma \ref{lem:LITTLEc3}:}
We will construct a path as follows.
We will construct 2 paths, $\hat{\bR}(t)$ and $\hat{\bR}'(t)$ starting from $\bR$ and $\bR'$,
respectively, ending in the same point. Then we may join them there, to get a chain that goes from $\bR$ to $\bR'$.

To begin with, 
we are going to perform the same algorithm for the two point-vectors $\bR$ and $\bR'$, independently of each other.
I.e., initially, the algorithm for $\bR$ will be done independently of what $\bR'$ happens to be, and vice-versa.
We call this initialization.

\underline{Initialization:}
The first part of the initialization step involves fixing for the fact that $\Lambda(d,N)$ may not be a perfect box,
since this will obstruct some later steps.
There are two possiblities due to the
construction of $\Lambda(d,N)$.
Either no point of $\Lambda(d,N)$ has the first coordinate equal to $L$.
Or else, for every point, if we replace its $d$th coordinate by $L$, then that new point is in $\B^d(L)$.
That is because of the lexicographic order we used.

In case 1, we perform the STEP2 algorithm from the proof of Lemma \ref{lem:LITTLEc3}, first on coordinate 1.
Note that on the first coordinate, every point already has that coordinate in $\{1,\dots,L-1\}$, as it also was in 
the proof of Lemma \ref{lem:LITTLEc3}.
After this, we have coordinate 1 for all points distinct.

Then we do the STEP1 algorithm for all coordinates $2$ to $d$. This moves all the coordinate away from $L$.
We do not cause any collisions because all points had the first coordinate distinct.
Then we do STEP2 from the proof of Lemma \ref{lem:LITTLEc3} for each coordinate $2$ to $d$, in turn.

In case 2, we perform the STEP2 algorithm from the proof of Lemma \ref{lem:LITTLEc3} on coordinate $d$, modified to replace $L$ by $L+1$
(because now the $d$th coordinate may be in $\{1,\dots,L\}$ instead of $\{1,\dots,L-1\}$.
Then we perform the STEP1 algorithm for each of the coordinates $1,\dots,d-1$.
This does not create any collisions because all points have distinct $d$th coordinate.
Then we perform STEP2 algorithm for each of the coordinates $1,\dots,d-1$, in turn.
Then we perform the STEP1 algorithm for the $d$th coordinate, in order to reduce its range to $\{1,\dots,L-1\}$.
This does not create any intersections because for each $j \in \{1,\dots,d-1\}$ all points have distinct $j$th coordinate.
Then we perform STEP2 again for coordinate $d$, to get all of the points to have distinct $d$th coordinate.

The initialization step is over.
Let us say that our new points are $\bR^{(2)}$ and $\bR^{(3)}$, in place of $\bR$ and $\bR'$, respectively.

\underline{Lining up coordinate sets:}
We point out that there is some bound $\mathrm{bd}(d,n)$ for the number of steps we have made, which depends only on $d$ and $n$.
For example, we may see that we could take $2d\mathfrak{N}''$.
This means that we have moved each coordinate of each point at most $\mathrm{bd}(d,n)$.
But initially, each coordinates of each point of $\bR$ is within $2\rho_{\max}(d,n)+1$ from the corresponding coordinate and point of $\bR'$.
So, now, what is true is that they are within $2\rho_{\max}(d,n)+1+2\mathrm{bd}(d,n)$.
We call this $\mathrm{bd}^{(1)}(d,n)$.

We now do the following.
For each coordinate index $j$, we order the set of coordinates of the points of $\bR^{(2)}$
and the set of coordinate of the points of $\bR^{(3)}$.
Then we move the former set to the latter set using the left-mover/right-mover algorithm to avoid collisions.
Note that each point moved at most $\mathrm{bf}^{(1)}(d,n)$ because those were the distances between the points initially, and ordering only makes distances decrease
(of the corresponding points, in the order).
This gives us a new $\bR^{(4)}$. For each coordinate index $j$, the set of $j$th coordinates of the points of $\bR^{(4)}$ is the same set as the $j$th coordinates
of the points of $\bR^{(3)}$.
But there may be a permutation, induced.

We note that we have now done at most $nd\mathrm{bd}^{(1)}(d,n)$ more steps.
Therefore,
\begin{equation}
\|\bR^{(4)} - \bR^{(3)}\|\, \leq\, (nd+1) \mathrm{bd}^{(2)}(d,n)\, \stackrel{\mathrm{def}}{=:}\, \mathrm{bd}^{(3)}(d,n)\, .
\end{equation}

\underline{Conclusion for $d=1$:}
If $d=1$, then we cannot necessarily get rid of the permutation.
In this case, we terminate. We permute all the point indices of the second point and path by the inverse of the permutation.
Then this will result in $\pi \bR'$ going to $\pi \bR^{(4)}$ and we choose $\pi$ so that $\pi \bR^{(4)} = \bR^{(3)}$.
Then doing the steps to go from $\bR$ to $\bR^{(3)}$ and then undoing the steps of the other chain goes from $\pi \bR^{(4)}$ to $\pi \bR'$.
This is what was claimed.

\underline{Conclusion for $d>1$:}

Coordinate, by coordinate fix the permutation.
If we are fixing the permutation in coordinate $j$, then this will cause collisions in coordinate $j$.
But since all the other coordinates have all distinct points, this does not cause us to enter $\mathscr{I}_n(\Lambda(d,N))$.
Draw the permutation and perform nearest neighbor transformations one at a time as they appear on the diagram, stretching the vertical direction
if necessary to remove degeneracies.
This will diminish distance at each stage because we are permuting to the correct order.
After a permutation is done in coordinate $j$, we end again with all distinct points.
Then move on to the next coordinate.
\end{proofof}

\bigskip
\appendix
\begin{center}
      {\LARGE  Appendices}
\end{center}

\section{Proof of advanced calculus fact for ``new lows''}
\label{app:Silly}

\begin{proofof}{\bf Proof of Lemma \ref{lem:realAn}:}
For any $\epsilon \in (0,1)$, we know that there exists $n_0$ such that for $n\geq n_0$, we have
$$
(1-\epsilon)Cn^{-p}\, \leq\, t_n\, \leq\, (1+\epsilon)C n^{-p}\, .
$$
We let $n_1$ be defined as 
$$
n_1\, =\, \min\{ n \in \{n_0,n_0+1,\dots\}\, :\,  (1-\epsilon)C n^{-p} < \min\{t_a,\dots,t_{n_0}\} \}\, .
$$
Note that $\min\{t_a,\dots,t_{n_0}\}$ is strictly positive, and $(1-\epsilon) C n^{-p}$ converges to $0$.
So there is an $n$ in $\{n_0,n_0+1,\dots\}$ such that 
$(1-\epsilon)C n^{-p} <\min\{t_a,\dots,t_{n_0}\}$.

Then we know that for all $n\geq n_1$ we have
$$
\min\{t_a,\dots,t_n\}\, \geq\, (1-\epsilon) C n^{-p}\, .
$$
On the other hand we know that for any $m \in \{n_0,n_0+1,\dots\}$,
$$
t_m\, \leq\, (1+\epsilon)C m^{-p}\, . 
$$
So defining $A(\epsilon;n)= \min\{k \in \Z\, :\, k \geq  [(1+\epsilon)/(1-\epsilon)]^{1/p}n \}$, which is an integer bigger than 1,
we know that for $n \in \{n_1,n_1+1,\dots\}$ and $m \in \{A(\epsilon;n),A(\epsilon;n)+1,\dots\}$,
$$
t_m\, \leq\, \min\{t_a,\dots,t_n\}\, .
$$
Now let $k^* = \min\{k \in \{a,a+1,\dots\}\, :\, t_k \leq t_{A(\epsilon;n)}\}$.
Note that with this choice we have $k^*\leq A(\epsilon;n)$. But by the last displayed equation for $m=A(\epsilon;n)$ we also have $k^*\geq n$.
Moreover, for any $k<k^*$ we have $t_k>t_{A(\epsilon;n)}\geq t_{k^*}$. So $t_{k^*} = \min\{t_a,\dots,t_{k^*}\}$.
So $k^*$ is a ``new low'' in the interval $\{n,\dots,A(\epsilon;n)\}$.

So $\mathcal{N}_{\mathrm{low}}(n) \leq A(\epsilon; n)$ since this is the smallest new low in $\{n,n+1,\dots\}$.

By hypothesis, we know that $t_{\mathcal{N}_{\mathrm{low}}(n)} \leq t_n \leq (1+\epsilon) C n^{-p}$.
But now we also know 
$$
t_{\mathcal{N}_{\mathrm{low}}(n)}\, \geq\, C (1-\epsilon) [\mathcal{N}_{\mathrm{low}}(n)]^{-p}\,
\geq\, C (1-\epsilon) [A(\epsilon;n)]^{-p}\, .
$$
Hence we see that
$$
(1-\epsilon)^2 (1+\epsilon)^{-1}\, =\, 
(1-\epsilon) \liminf_{n \to \infty} [A(\epsilon;n)/n]^{-p}\, 
\leq\, \liminf_{n \to \infty} \frac{t_{\mathcal{N}_{\mathrm{low}}(n)}}{C n^{-p}}\,
\leq\, \limsup_{n \to \infty} \frac{t_{\mathcal{N}_{\mathrm{low}}(n)}}{C n^{-p}}\,
\leq\, 1+\epsilon\, .
$$
Since $\epsilon \in (0,1)$ is arbitrary, and the bracketing quantities both converge to $1$ as $\epsilon \to 0$,
this proves the claim.
\end{proofof}

\section{Proof of summary of spectrum of $\mathfrak{H}^{(n)}_{\B^d(L)}$}
\label{app:specSummary}

The full spectrum of $\mathfrak{H}^{(n)}_{\B^d(L)}$ may be calculated as products of functions $\prod_{k=1}^{n}\prod_{j=1}^{d} f(L^{-1} \kappa_{k,j},r_{k,j})$,
unsymmetrized.
This may be seen in $d=1$. But then it follows from separation of variables for higher $d$.
Then symmetrization is what is necessary to restrict to the range of $\mathfrak{S}^{(n)}_{\B^d(L)}$.

\section{Proofs of easy facts about the graph Laplacians}
\label{app:Inter}

We leave these lemmas as an exercise. They can be done using the fact that $H_{\mathscr{G}} \restriction \Hil_{\mathscr{V}}^{(n)}$
is unitarily equivalent to $\widetilde{\mathfrak{H}}^{(n)}_{\mathscr{G}}$ restricted to the range of $\mathfrak{S}^{(n)}_{\mathscr{V}}$,
and properties of graph Laplacians: such as monotonicity relative to the psd cone under the action of adding edges.


\section{Proof of Proposition \ref{prop:Variational1}}
\label{app:AsymptoticEnergy}

Based on the results of Appendix \ref{app:Inter}, what we must prove is the following.
Suppose that we have $(\k_1,\dots,\k_n) \in (\{0,1,\dots\}^d)^n$.
For some $L_0$, we have $(\k_1,\dots,\k_n) \in (\{0,\dots,L_0-1\}^d)^n$.
Then we consider only values of $L$ with $L\geq L_0$.
We define $\nu_{(\k_1,\dots,\k_n)} \in \mathscr{O}_n(d,L)$.
Then, considering $\widetilde{F}^{(n)}_{d,L}(\nu_{(\k_1,\dots,\k_n)};\cdot)$,
we know by Lemma \ref{lem:specSummary} that
$$
\mathfrak{H}^{(n)}_{\B^d(L)} \widetilde{F}^{(n)}_{d,L}(\nu_{(\k_1,\dots,\k_n)};\cdot)\, 
=\, \lambda^{(n)}_{d,L}(\k_1,\dots,\k_n)  \widetilde{F}^{(n)}_{d,L}(\nu_{(\k_1,\dots,\k_n)};\cdot)\, ,
$$
where
\begin{equation}
\label{eq:lambdaDEF}
\lambda^{(n)}_{d,L}(\k_1,\dots,\k_n)\, =\, \sum_{k=1}^{n} \sum_{j=1}^{d} 2 \sin^2\left(\frac{\pi \kappa_{k,j}}{2L}\right)\, .
\end{equation}
Therefore, for any $(\r_1,\dots,\r_n) \in \Lambda(d,N)^n \setminus \mathscr{I}_n(\Lambda(d,N))$,
we have
\begin{multline*}
\left(\widetilde{\mathfrak{H}}^{(n)}_{\Lambda(d,N)}
- \lambda^{(n)}_{d,L}(\k_1,\dots,\k_n)\right) 
\widetilde{F}^{(n)}_{d,L}(\nu_{(\k_1,\dots,\k_n)};\r_1,\dots,\r_n)\\
=\, 
\left(\widetilde{\mathfrak{H}}^{(n)}_{\Lambda(d,N)}
- \mathfrak{H}^{(n)}_{\B^d(L)}\right) 
\widetilde{F}^{(n)}_{d,L}(\nu_{(\k_1,\dots,\k_n)};\r_1,\dots,\r_n)\, .
\end{multline*}
Since everything is symmetric, this implies that
\begin{multline*}
\left\|\left(H_{\Lambda(d,N)}-\lambda^{(n)}_{d,L}(\k_1,\dots,\k_n)\right) \widetilde{\Psi}^{(n)}_{d,N}(\k_1,\dots,\k_n)\right\|^2\\
=\, \sum_{(\r_1,\dots,\r_n) \in \Lambda(d,N)^n \setminus \mathscr{I}_n(\Lambda(d,N))}
\left|\left(\widetilde{\mathfrak{H}}^{(n)}_{\Lambda(d,N)}
- \mathfrak{H}^{(n)}_{\B^d(L)}\right) 
\widetilde{F}^{(n)}_{d,L}(\nu_{(\k_1,\dots,\k_n)};\r_1,\dots,\r_n)\right|^2\, .
\end{multline*}
We may also rewrite this formula as 
\begin{multline*}
\left(\widetilde{\mathfrak{H}}^{(n)}_{\Lambda(d,N)}
- \mathfrak{H}^{(n)}_{\B^d(L)}\right) 
\widetilde{F}^{(n)}_{d,L}(\nu_{(\k_1,\dots,\k_n)};\r_1,\dots,\r_n)\\
=\, 
\frac{1}{2}\, \sum_{\substack{(\r_1',\dots,\r_n') \\
\|(\r_1',\dots,\r_n') - (\r_1,\dots,\r_n)\|=1}} \left(1-\mathbf{1}_{\Lambda(d,N)^n\setminus \mathscr{I}_n(\Lambda(d,N))}(\r_1',\dots,\r_n')\right)
(F(\r_1,\dots,\r_n)-F(\r_1',\dots,\r'_n))\, ,
\end{multline*}
where we have written $F$ for $\widetilde{F}^{(n)}_{d,L}(\nu_{(\k_1,\dots,\k_n)};\cdot)$ in this formula to reduce notation.
The number of terms in this sum is always bounded by $2nd$ (the degree of $(\Z^d)^n \cong \Z^{nd}$).
Therefore, by Cauchy-Schwarz, and a sup-norm bound,
we may bound
\begin{multline*}
\left\|\left(H_{\Lambda(d,N)}-\lambda^{(n)}_{d,L}(\k_1,\dots,\k_n)\right) \widetilde{\Psi}^{(n)}_{d,N}(\k_1,\dots,\k_n)\right\|^2\\
\leq\, nd \max_{\substack{(\r_1,\dots,\r_n), (\r_1',\dots,\r_n') \in (\B^d(L))^n\\ 
\|(\r_1',\dots,\r_n') - (\r_1,\dots,\r_n)\|=1}} 
|F(\r_1,\dots,\r_n)-F(\r_1',\dots,\r'_n)|^2 \cdot \mathcal{N}\, ,
\end{multline*}
where $\mathcal{N}$ equals the number of edges in $\Phi_n(\B^d(L)) \setminus \Theta_n(\Lambda(d,N))$.
For any pair $(\r_1,\dots,\r_n), (\r_1',\dots,\r_n') \in (\B^d(L))^n$ with $\|(\r_1',\dots,\r_n') - (\r_1,\dots,\r_n)\|=1$,
we note, by the fundamental theorem of calculus (and Cauchy-Schwarz), that
$$
|F(\r_1,\dots,\r_n)-F(\r_1',\dots,\r'_n)|^2\, \leq\, \int_{0}^{1} \Big\|\frac{d}{dt} F(t\r_1+(1-t)\r_1',\dots,t\r_n+(1-t)\r_n')\Big\|^2 dt\, ,
$$
where we exted the definition of $F$ from $(\B^d(L))^n$ to $\R^{dn}$ by the same formula as in  Definition \ref{def:WidetildeF}.
We note that only one coordinate $k \in \{1,\dots,n\}$ has $\r_k\neq \r_k'$, and in this case there is only one $j \in \{1,\dots,d\}$ such that
$r_{kj} \neq r'_{kj}$.
Then direct calculation shows
$$
\Big\|\frac{d}{dt} F(t\r_1+(1-t)\r_1',\dots,t\r_n+(1-t)\r_n')\Big\|^2\, \leq\, 
L^{-nd} |\mathfrak{K}_{d,L}(\nu)|^{-1} 2^{nd}\, \cdot \frac{\pi^2 \kappa_{kj}^2}{L^2}\, .
$$
So, what we see is that
$$
\left\|\left(H_{\Lambda(d,N)}-\lambda^{(n)}_{d,L}(\k_1,\dots,\k_n)\right) \widetilde{\Psi}^{(n)}_{d,N}(\k_1,\dots,\k_n)\right\|^2\\
\leq\, nd2^{nd} \max_{\substack{k \in \{1,\dots,n\}\\ j \in \{1,\dots,d\}}} \frac{\kappa_{kj}^2}{L^2} \cdot \frac{\mathcal{N}}{L^{nd}}\, .
$$
But we claim  that there is some constant $C(n,d)$ such that
$$
\frac{\mathcal{N}}{L^{nd}}\, \leq\, \frac{C(n,d)}{L}\, ,
$$
because at least one of the points $(\r_1,\dots,\r_n), (\r_1',\dots,\r_n') \in (\B^d(L))^n$ must be in $(\B^d(L))^n \setminus \mathcal{I}_n(\Lambda(d,N))$.
Assume it is $(\r_1',\dots,\r_n')$.
Then this means that at least one of the two scenarios must occur. 
The first possibility is that for some $k \in \{1,\dots,d\}$ we have $\r_k' \in \B^d(L) \setminus \Lambda(d,N)$, meaning that for some $j \in \{1,\dots,d\}$
we must have $r'_{kj} = L$. This restricts that choice of that coordinate.
Or, the second possibility is that for some pair $1\leq k_1<k_2 \leq n$ we have $\r_{k_1}' = \r_{k_2}'$.
This restricts all $j$ coordinates $r_{k_2,j}'$ once we have chosen $\r_{k_1}'$.
So this proves the claim.
Putting this all together, we obtain
$$
\left\|\left(H_{\Lambda(d,N)}-\lambda^{(n)}_{d,L}(\k_1,\dots,\k_n)\right) \widetilde{\Psi}^{(n)}_{d,N}(\k_1,\dots,\k_n)\right\|^2\\
\leq\, C'(d,n) \max_{\substack{k \in \{1,\dots,n\}\\ j \in \{1,\dots,d\}}} \frac{\kappa_{kj}^2}{L^3}\, .
$$
Multiplying through by $\gamma^{-1} \cdot L^2$, we obtain
$$
\left\|\left(\gamma^{-1} \cdot L^2 H_{\Lambda(d,N)}-\gamma^{-1} \cdot L^2 \lambda^{(n)}_{d,L}(\k_1,\dots,\k_n)\right) \widetilde{\Psi}^{(n)}_{d,N}(\k_1,\dots,\k_n)\right\|^2\\
\leq\, C'(d,n) \max_{\substack{k \in \{1,\dots,n\}\\ j \in \{1,\dots,d\}}} \frac{\kappa_{kj}^2}{L}\, .
$$
The right hand side does converge to $0$ as $L \to \infty$.
Moreover, by (\ref{eq:lambdaDEF}), we do have the simple formula
$$
\lim_{L \to \infty} L^2 \lambda^{(n)}_{d,L}(\k_1,\dots,\k_n)\,
=\, \lim_{L \to \infty} \sum_{k=1}^{n} \sum_{j=1}^{d} 2L^2 \sin^2\left(\frac{\pi \kappa_{k,j}}{2L}\right)\, 
=\, \sum_{k=1}^{n} \sum_{j=1}^{d} \frac{\pi^2}{2}\, \kappa_{k,j}^2\, 
=\, \gamma \sum_{k=1}^{n} \sum_{j=1}^{d} \kappa_{k,j}^2\, .
$$
This proves (\ref{eq:AsymptoticEnergy}).

In principle it is easier to prove (\ref{eq:AsymptoticIP}) since one is operating at the level of $\ell^2$ not energy.
In fact, the main difference is that there will not be any need for the fundamental theorem of calculus.
Everything follows from sup-norm bounds on the functions $f$ (which is akin to a type of near equipartition of the $\ell^2$ energy)
along with the fact that the relative fraction of bad set of vertices is still bounded by $\mathcal{N}/L^d = O(L^{-1})$.
We leave the details to the reader.

\section*{Acknowledgments}
This research was partially supported by the National Science Foundation under Grant DMS-1515850 (B.N.).

\end{document}